\documentclass[a4paper,onecolumn,superscriptaddress,11pt,%
        allowfontchangeintitle,allowtoday%
                ]{quantumarticle}
\pdfoutput=1

\usepackage{amsmath}
\usepackage{enumerate}
\usepackage{amssymb}
\usepackage{amsfonts}
\usepackage{amsthm}
\usepackage{mathtools}
\usepackage[usenames]{color}
\usepackage{bm}
\usepackage{dsfont} 
\usepackage[numbers,sort&compress]{natbib}
\usepackage{accents}
\usepackage{hyperref}
\theoremstyle{definition}
\newtheorem{theorem}{Theorem}
\newtheorem{proposition}[theorem]{Proposition}

\newtheorem{lemma}[theorem]{Lemma}
\newtheorem{corollary}[theorem]{Corollary}
\newtheorem{definition}[theorem]{Definition}
\newtheorem{example}[theorem]{Example}

\newtheorem*{theorem*}{Theorem}
\theoremstyle{remark}
\newtheorem{remark}[theorem]{Remark}

\usepackage{scalerel, stackengine}
\stackMath
\newcommand\reallywidehat[1]{%
\savestack{\tmpbox}{\stretchto{%
  \scaleto{%
    \scalerel*[\widthof{\ensuremath{#1}}]{\kern-.6pt\bigwedge\kern-.6pt}%
    {\rule[-\textheight/2]{1ex}{\textheight}}
  }{\textheight}%
}{0.5ex}}%
\stackon[1pt]{#1}{\tmpbox}%
}
\newcommand\reallywidecheck[1]{%
\savestack{\tmpbox}{\stretchto{%
  \scaleto{%
    \scalerel*[\widthof{\ensuremath{#1}}]{\kern-.6pt\bigvee\kern-.6pt}%
    {\rule[-\textheight/2]{1ex}{\textheight}}
  }{\textheight}%
}{0.5ex}}%
\stackon[1pt]{#1}{\tmpbox}%
}
\newcommand\reallywidetilde[1]{%
\savestack{\tmpbox}{\stretchto{%
  \scaleto{%
    \scalerel*[\widthof{\ensuremath{#1}}]{\kern-.6pt\sim\kern-.6pt}%
    {\rule[-\textheight/2]{1ex}{\textheight}}
  }{\textheight}%
}{0.5ex}}%
\stackon[1pt]{#1}{\tmpbox}%
}

\newcommand{\R}{{\mathbb{R}}}
\renewcommand{\H}{{\mathbb{H}}}

\newcommand{\beq}{\begin{equation}}
\newcommand{\eeq}{\end{equation}}
\newcommand{\beqa}{\begin{eqnarray}}
\newcommand{\eeqa}{\end{eqnarray}}
\newcommand{\beqas}{\begin{eqnarray*}}
\newcommand{\eeqas}{\end{eqnarray*}}

\newcommand{\ket}[1]{| #1 \rangle}

\newcommand{\ketbra}[2]{| #1 \rangle\langle #2 |}

\newcommand{\omax}{{\otimes_{max}}}
\newcommand{\omin}{{\otimes_{min}}}

\newcommand{\one}{{\mathds{1}}}

\providecommand{\cT}{{\cal T}}

\renewcommand{\phi}{\varphi}

\def\textsec{\textsection}

\def\intersect{\cap}

\def\face{{\rm Face}}

\def\homega{\hat{\omega}}

\def\implies{\Rightarrow}

\def\id{{\rm id}}

\def\R{{\mathbb{R}}}
\def\C{{\mathbb{C}}}

\newcommand{\tr}{\text{tr\,}}

\newcommand{\aut}{\text{Aut}}

\newcommand{\notiso}{{\ncong}}

\DeclarePairedDelimiter{\inner}{\langle}{\rangle}
\DeclarePairedDelimiter{\inn}{\langle}{\rangle}

\newcommand{\fk}{\mathfrak{k}}

\newcommand{\tempout}[1]{{}}

\newcommand{\added}[1]{#1}

\newcommand{\myheading}[1]{\paragraph{#1}}

\begin{document}
\title{Self-duality and Jordan structure of quantum theory follow from homogeneity and pure transitivity}

\author{Howard Barnum}
\email{hnbarnum@aol.com}
\affiliation{Center for Quantum Information and Control, Dept of Physics and Astronomy, University of New Mexico, Albuquerque NM}
\author{Cozmin Ududec}
\email{cozster@gmail.com}
\affiliation{Unaffiliated}
\author{John van de Wetering} 
\email{j.m.m.vandewetering@uva.nl}
\affiliation{University of Amsterdam, Dept of Computer Science, Amsterdam, The Netherlands}
\thanks{Corresponding author}

\date{\today}

\begin{abstract}
Quantum theory satisfies the mathematically powerful property of \emph{self-duality}, meaning that its state and effect spaces are made isomorphic by an inner product. This is why both quantum states and effects can be represented by positive Hermitian operators. The seminal Koecher-Vinberg theorem shows that self-duality is a rather special property, as it states that any system that is self-dual and also \emph{homogeneous}, meaning that the group of symmetries of its cone of unnormalised states acts transitively on 
the strictly positive states, must be isomorphic to a \emph{Euclidean Jordan algebra} (EJA). EJAs have been classified and are known to be just a slight generalisation of quantum systems. This is why the Koecher-Vinberg theorem, and hence the properties of self-duality and homogeneity, 
have been used in several reconstructions of quantum theory from first principles.
While homogeneity has an operational derivation from the ability to \emph{steer} states, self-duality currently lacks such a clear operational motivation.
In this paper we prove 
an alternative to the Koecher-Vinberg theorem, substituting for self-duality the more operational property of 
\emph{pure transitivity}:  that the symmetries of the normalised state space act transitively on the pure states.  This property has been extensively used in work on the informational foundations of quantum theory, quantum information, and quantum statistical mechanics.         
We show that any system that is homogeneous and satisfies pure transitivity  
must be self-dual, and hence isomorphic to an EJA.  Together with various (mostly known) 
ways of singling out the complex matrix algebras from among EJAs, and known ways of ruling out classicality in these, this yields several new and 
concise reconstructions of quantum theory.  For example, quantum systems, classical systems, and composites of these are the only ones that are homogeneous, satisfy pure transitivity, and allow locally tomographic composites; fully quantum systems are the only ones that are homogeneous, have continuous pure transitivity, and allow a correspondence between observables and generators of reversible transformations.
\end{abstract}

\maketitle

\section{Introduction}


Two fundamental concepts for an operational description of a physical system are \emph{states} and \emph{effects}. Roughly speaking, the states of a system are all the ways in which a system can be prepared. The effects correspond to all the possible outcomes of the various ways in which a system can be measured.  The geometry of states and the geometry of effects constrain each other, and also constrain the possible transformations, i.e.~dynamical evolutions, the system may undergo.  


\added{
The study of fundamental physical and information-processing properties of the quantum formalism was revitalised in the 1990s as part of the rapid growth of the field of quantum information processing. 
Much of this was done with a general framework of states, effects and transformations in what are now often called \emph{general probabilistic theories} (GPTs)~\cite{BarrettGenProb}. 
Particular geometric properties have been shown to be necessary and/or sufficient for certain information-processing or physical protocols or tasks in a general physical theory, while other properties not possessed by ordinary quantum theory have been shown to enable certain tasks not possible (or not efficiently implementable) within standard quantum theory.   Examples come from distributed computation 
\cite{VanDamImplausible}, 
quantum random access coding 
\cite{NayakRandomAccess, PawlowskiInformationCausality,VerSteegWehnerRelaxed}, oblivious transfer \cite{WolfWullschlegerOblivious},  remote state preparation 
\cite{CavalcantiPostQuantumSteering}, and many other areas.  Treating physical constraints more abstractly is also useful in quantum cryptography, which frequently uses a state-effect framework to describe the actions available to an adversary in proofs of information-theoretic security or insecurity; see e.g.~\cite{Lo97a, Mayers97a, DArianoBitCommitment, Lo98a, Lo99a, Deutsch96a}.
Beyond just being of theoretical interest, fundamental inquiry into the information processing and physical possibilities of general probabilistic theories can play a part in finding limits on how quantum and general relativistic physics can be reconciled.
}

A particularly nice property of quantum theory is that we can treat states and effects on the same footing: the space of unnormalised quantum states corresponds to the Hermitian operators on a Hilbert space, while the space of unnormalised effects \emph{also} corresponds to the Hermitian operators. In particular, if we have a pure state $\ketbra{\psi}{\psi}$, the effect that tests for it is given by exactly the same operator $\ketbra{\psi}{\psi}$. This correspondence between states and effects is given by the inner product $\inn{A,B} = \tr(AB)$ that implements the Born probability $\tr(\rho E)$ for a state $\rho$ and effect $E$. 

This property that the sets of (unnormalised) states and effects are identified with each other by a probability-determining inner product is called 
\emph{self-duality}.  Among hypothetical physical theories, this is quite `rare'.
In fact, the celebrated Koecher-Vinberg theorem \cite{Koecher,VinbergHomogeneousCones} states that any system described by an ordered vector space of unnormalised states that is both self-dual and \emph{homogeneous} is a \emph{Euclidean Jordan algebra}, an algebraic generalisation of a standard quantum system. Homogeneity 
states that the ordered vector space is maximally symmetric, with all strictly positive states (full-rank states in the quantum case) connected by order isomorphisms.

Homogeneity is closely linked to the ability to \emph{steer} arbitrary states, which gives it an operational interpretation~\cite{BarnumSteer}. However, while self-duality is an interesting mathematical property with rich consequences, its operational meaning is currently not clear.  
In~\cite{MullerSD2012} it was shown that the more operational property of \emph{bit symmetry} implies (but is not equivalent to) self-duality; some conjunctions of strong but appealing operational  conditions implying (but again not equivalent to) self-duality were given in \cite{WilceRoyalRoad}, in a derivation of quantum theory via the Koecher-Vinberg theorem.  
In~\cite{BBLW08} it was shown that \emph{weak} self-duality is necessary for state teleportation through a copy of the teleported system. Weak self-duality states that the state and effect space are isomorphic.  However, this isomorphism is not necessarily given by a  inner product.  

Given that the Koecher-Vinberg theorem has proven so useful, it would be desirable to either give self-duality an operational interpretation or otherwise replace it with an alternative more operational property.
In this paper we show that homogeneity in combination with the property of \emph{pure transitivity} implies self-duality.  Pure transitivity states that any two pure states can be mapped into each other by a reversible transformation, and hence has the clear operational interpretation that any two pure states of a system should be reversibly dynamically transformable into each other. 
\added{
In particular, we can imagine the following operational protocol: an adversary picks two pure states of the system. You must then find a pure and reversible transformation that maps the first state into the second. A system satisfies pure transitivity if and only if you can always win this game. The properties of purity and reversibility in themselves can then also be operationally verified.
}
 
Pure transitivity, and its slightly stronger version, \emph{continuous} pure transitivity, are assumptions often used 
(cf.~e.g.~\cite{Hardy2001a,MasanesAxiom,MasanesEtAlEntanglementAndBlochBall}) in the theory of 
GPTs.
 Pure transitivity also follows from \emph{essential uniqueness of purification} in combination with \emph{pure conditioning} (that pure measurements preserve pure states)~\cite{GiulioDer11}, two other assumptions often made in the field of reconstructions of quantum theory.
Continuous paths of pure states are important in statistical
mechanics---for example, the existence of such a path between any pair of pure 
states is crucial to von Neumann's derivation of the expression for quantum entropy \cite{vonNeumannBook}. 

Our main result is the following, which can be interpreted as an an alternative to the Koecher-Vinberg theorem: 
\begin{theorem*}
    An ordered vector space that is homogeneous and satisfies pure transitivity is order-isomorphic to a Euclidean Jordan algebra.
\end{theorem*}
This result hence gives a derivation of self-duality from arguably more operational assumptions.  But regardless of one's view on the relative merits and interest of pure transitivity and self-duality, their equivalence in the context of homogeneity greatly illuminates the nature of 
quantum information and the geometry of its state spaces, while providing an alternative and at least equally compelling derivation of Jordan structure.   

Our proof of this theorem is an application of Vinberg's theory of homogeneous cones~\cite{VinbergHomogeneousTheory} and their automorphism groups~\cite{VinbergStructure}, in combination with some results from~\cite{TruongTuncel}.

Using the well-established fact  that the only Jordan algebras that allow locally-tomographic composites are the standard quantum systems, this gives us a new way to characterise quantum theory: it is the only theory of systems that are homogeneous, satisfy pure transitivity and have locally tomographic composites.  

Mildly generalizing results from \cite{BarnumSteer} concerning Schr\"odinger's notion of \emph{steering} 
allows us to substitute \emph{universal self-steering} (the steerability of each system in the theory by a copy of itself), or even universal steerability by a system of the same dimension, for homogeneity in some of our results: 
\begin{theorem*}
    (informal) Universal self-steering, continuous pure transitivity, and local tomography characterize quantum theory.
\end{theorem*}
Hence, this characterisation gives a way to understand quantum systems through the lens of certain protocols we can apply to them: we can always steer states on entangled systems, we can always continuously transform a pure state into any other pure state, and we can fully characterise an entangled system using local measurements.


\section{Preliminaries}

In order to state our results formally, we need to recall a number of concepts from generalised probabilistic theories and ordered vector spaces.

\subsection{Generalised probabilistic theories and ordered vector spaces}

In a finite-dimensional generalised probabilistic theory (GPT) we have for every system a state space $\Omega$ and an effect space $E$, which are both convex sets~\cite{BarrettGenProb}.

We additionally have a probability function: for each state $\omega \in \Omega$ and $e\in E$ we assign a probability $\inner{e,\omega}\in [0,1]$ in a manner which respects the convex structure\added{, i.e. 
which is affine in each argument (``biaffine'').}
We will assume that the states are separated by the effects, meaning that if $\inner{e,\omega_1} = \inner{e,\omega_2}$ for all effects $e\in E$, then $\omega_1=\omega_2$. 
We can then take the vector space $V$ of formal sums of states modulo equality of all effects, so that we can view $\Omega\subseteq V$. We will also assume that effects are separated by states, so that $E$ spans the dual space $V^*$.  As is standard, we will assume that $\Omega$ and $E$ are closed sets,
and hence \added{(since they are necessarily bounded by the probability function)} compact. 
We will further assume there is a unique \emph{unit} effect $u$ satisfying $\inner{u,\omega} = 1$ for all $\omega\in \Omega$.
We can then define a (finitely indexed) \emph{measurement} as an indexed set of effects $e_i$, $i \in \{1,..n\}$ such 
that $\sum_i e_i = u$.  This
generalizes the notion of a  positive operator-valued measure (POVM) (on 
the $\sigma$-algebra of subsets  of $\{1,...,n\}$) from quantum theory.  
(As in quantum theory, effect-valued measures on more general $\sigma$-algebras may also be defined \cite{Edwards} but we do not need them here.)  
We will also assume on $E$ the condition that
for every $e \in E$,  also $u-e \in E$.  This ensures that every effect
is part of at least one measurement, consisting of $e_1 := e, ~ e_2 := u-e$.
\added{Finally, we will assume that $V$ is finite-dimensional.  This is a standard assumption in most of the literature on generalised probabilistic theories and also has an operational formulation as the condition of \emph{finite tomography}: that we can fully characterize any state using a finite number of distinct measurement outcomes.}
Let's summarise this in a formal definition.
\begin{definition}\label{def:GPT}
    A \emph{Generalised probablistic theory} (GPT) system $(\Omega, E, \inner{\cdot, \cdot})$ consists of 
    \begin{itemize}
        \item A state space $\Omega \subseteq V$ that is a closed convex subset of the real \added{finite-dimensional} vector space $V$.
        \item An effect space $E\subseteq V^*$ that is a closed convex set which spans $V^*$.
        \item A probability function $\inner{\cdot,\cdot}:E\times \Omega\to [0,1]$ that assigns to every state $\omega\in\Omega$ and effect $e\in E$ a probability $\inner{e,\omega} \in [0,1]$.
    \end{itemize}
    These are required to satisfy the following conditions:
    \begin{itemize}
        \item There is a unique unit effect $u\in E$ satisfying $\inn{u,\omega}=1$ for all $\omega\in\Omega$.
        \item $\inner{\cdot,\cdot}$ is affine in both arguments (hence respects convex mixtures).
        \item For all $e\in E$, $u-e\in E$.
        \item If $\inner{e,\omega_1}=\inner{e,\omega_2}$ for all $e\in E$, then $\omega_1=\omega_2$.
    \end{itemize}
 \end{definition} 

This basic formalism of states and effects, including the infinite-dimensional case, traces back at least to \cite{Ludwig64a, Ludwig67a} (see also \cite{LudwigAx}).

 \begin{example}
     The GPT of quantum theory has, for each positive integer $n$, systems corresponding to the vector space $V=M_n(\mathbb{C})_{\text{sa}}$ of self-adjoint $n\times n$ complex matrices. The state space $\Omega\subseteq V$ corresponds to density operators $\omega \in \Omega \iff \omega\geq 0, \tr(\omega)=1$. The effect space $E$ corresponds to $e\in V$ that are positive and subunital: $0\leq e\leq \one$ (we write $\one$ for the identity operator).  The probability function is given by $\inner{e,\omega} = \tr(e\omega)$.
 \end{example}

The space $V$ generated by the states $\Omega$ of a GPT system is an \emph{ordered vector space}:
\begin{definition}
    An \emph{ordered vector space} is a real vector space $V$ equipped with a partial order $\leq$ such that:
    \begin{itemize}
        \item It respects addition (if $a\leq b$, then $a+c\leq b+c$ for all $a,b,c\in V$).
        \item It respects positive scalar multiplication (if $a\leq b$, then $\lambda a \leq \lambda b$ for all $a,b\in V$ and $\lambda \geq 0$).
    \end{itemize}
 \end{definition}
The partial order of an ordered vector space is completely determined by the cone of positive elements $V_+ := \{a \in V\,|\, a\geq 0\}$.  A \emph{cone} in a real vector space $V$ is a subset $C$ closed under addition and nonnegative scalar multiplication.  It is called \emph{pointed} if 
${C}\cap  -{C} = \{0\}$.
The positive cone of an ordered vector space is pointed, and every such cone $C$  
defines an ordered vector space
 \added{for which $V_+ = C$, by defining the relation $a \ge b$ as $a - b \in C$}.  
For a GPT system, the positive cone consists of the \emph{unnormalised} states: nonnegative multiples of elements of $\Omega$. 
In quantum theory this is the cone of positive-semidefinite matrices.

\begin{remark}
    \added{A cone is \emph{generating} if it spans $V$; we call a cone \emph{regular} if it closed, pointed and generating.  The positive cone of an ordered vector space is not necessarily generating, but our assumption that effects separate states ensures that the positive cone of a GPT system is.} 

\end{remark}

In finite dimension, a real vector space has a unique topology compatible with the 
linear structure (given for example by the linear isomorphism $V\cong \R^n$ and 
the usual topology thereon) so we can speak of the \emph{interior} of the cone. 
For a GPT, states in the interior correspond to \emph{strictly positive} states $\omega$, 
i.e.~which satisfy $\inner{e,\omega} > 0$ for all nonzero effects $e\in E$.
In an operational setting we could never hope to eliminate all noise to ensure that
some measurement outcome should occur with zero probability. Interior states are, in this sense, adequate to describe observable phenomena.   By the same token, adjoining the boundary by working with closed cone or convex set of states or effects is operationally harmless, and can be very useful mathematically.  

\begin{example}
     In the GPT of quantum theory the order on the space $V=M_n(\mathbb{C})_{\text{sa}}$, 
the self-adjoint $n\times n$ complex matrices, is given by the standard L\"owner order where $A\geq 0$ iff $A$ is positive semi-definite. The positive interior elements are precisely those positive operators that are full-rank, or equivalently (in finite-dimension) are invertible.
 \end{example}

\myheading{Pure and normalised states}%
If $V$ is ordered, we can order its dual $V^*$ by setting $f\leq g$ iff $\inner{f,a}\leq \inner{g,a}$ for all $a\geq 0$ in $V$. We can therefore find a functional $u\in V^*_+$ that is in the interior of the dual cone.  Using any such functional we can define a set of `normalised states' in $V$ as $(V_+)_u := \{a \geq 0 \,|\, \inner{a,u} = 1\}$.  These are precisely the \emph{bases} of the cone $V_+$: compact convex subsets $B$ 
of $V$ such that each nonzero element of $V_+$ is uniquely expressed as a nonnegative multiple of
an element of $B$.  
For a GPT system we recover the standard notion of normalised state if we take the functional $u$ to be the unique unit effect, i.e. the functional whose constant value on $\Omega$ is $1$.


In the case of quantum theory, the functional $u$ is usually chosen to be the trace 
, and the corresponding effect as an Hermitian operator is the identity operator $\one$. 

\begin{definition}\label{def: ordered vector space}
    Let $V$ be an ordered vector space with chosen functional $u$. The \emph{pure states} of $(V,u)$ are the extremal elements of the convex set $(V_+)_u$, i.e.~the elements that cannot be expressed as nontrivial convex combinations of other elements of the set $(V_+)_u$.
\end{definition}

This definition of pure state matches the one for GPTs.
However, in a GPT one usually defines a \emph{pure effect} to be one which cannot be ``coarse-grained''; that is, which cannot be expressed as a sum of nonzero effects that are not proportional to each other.  We will adopt this definition.  It is not necessarily the same as extremality in the convex set $E$ of effects.

\begin{example}
    In the GPT of quantum theory with a system $V=M_n(\mathbb{C})_{\text{sa}}$ and chosen functional the trace, the normalised states are those that satisfy $1 = \inner{\one,\omega} = \tr( \one \omega) = \tr(\omega)$, and hence are those that have trace 1. The base $(V_+)_\tr$ hence corresponds to the positive, trace 1 operators, i.e.~the density matrices.  The pure states, the convex extremal elements of $(V_+)_\tr$, are then precisely the rank 1 operators $\ketbra{\psi}{\psi}$. An effect $0\leq e\leq \one$ cannot be coarse-grained if it is a rank 1 operator, i.e.~is proportional to a $\ketbra{\psi}{\psi}$.  
\end{example}

\begin{definition}
    For ordered vector spaces $V_1$ and $V_2$ we define their \emph{direct sum} $V_1 \oplus V_2$ as the direct sum of vector spaces, with the order defined pointwise: $x \oplus x_2 \ge y_1 \oplus y_2$ iff $x_1 \ge y_1 $ and $x_2 \ge y_2$. We say an ordered vector space $V$ is \emph{reducible} if $V = V_1 \oplus V_2$ for some ordered $V_1, V_2 \ne \{0\}$. We instead call it \emph{irreducible} if no such decomposition exists.
\end{definition}

Every ordered linear space has an essentially unique (i.e., unique up to permutation of the 
labels $i$) decomposition $V = \oplus_i V_i$ into irreducible spaces.  
If $V = V_1 \oplus V_2$, then for some (and hence, every) base $\Omega$ of the cone, all pure states of $\Omega$ are in either $V_1$ or $V_2$.
One may hence think of these summands as ``superselection sectors'' and information 
about which summand of the cone a state is in as classical information. 
This information can be extracted using the ``classical effects'' $(u_1,0)$ and $(0,u_2)$ that assign a probability 1 to all states fully in one part of the decomposition.

\myheading{Order-preserving maps}%
In a GPT, transformations between systems map the state spaces to each other, while preserving the convex structure (i.e., commuting with convex combination).   On the level of the corresponding ordered spaces, these transformations correspond to \emph{positive maps}.
\begin{definition}
    Let $V$ and $W$ be ordered vector spaces. We say a linear map $f: V \rightarrow W$ is \emph{positive} when $a \le b \implies f(a) \le f(b)$, or equivalently $\Phi(V_+) \subseteq W_+$. An \emph{order isomorphism} $\Phi:V\to W$ is a linear bijection satisfying $a \leq b \iff \Phi(a)\leq \Phi(b)$, and we furthermore call it an \emph{order automorphism} if $V=W$.
\end{definition}

It is immediate from this definition that an order isomorphism is the same thing as a positive linear map with positive inverse.  

A positive map $f:V\to W$ determines a positive map $f^*:W^*\to V^*$ via $f^*(e)(a) = e(f(a))$, called $f$'s \emph{dual}.  It is an order-isomorphism if and only if $f$ is.  
Given a choice $u_V, u_W$ of units, we say a positive map is \emph{normalised} when $f^*(u_W) = u_V$, 
i.e.~when $f((V_+)_u) \subseteq (W_+)_u$,
and \emph{sub-normalised} when $f^*(u_W) \le u_V$.
State transformations in a GPT correspond to normalised positive maps and reversible transformations to normalised order isomorphisms.
\begin{example}
    In the GPT of quantum theory with a system $V=M_n(\mathbb{C})_{\text{sa}}$ and chosen functional $\one$, the positive maps are those that send positive semi-definite operators to positive semi-definite operators. They are normalised when they are trace-preserving, and sub-normalised when they are trace-non-increasing. The connected identity component of the group of order automorphisms is the set of all transformations  of the form $\omega \mapsto A\omega A^\dagger$ for some invertible $A$ where $A^\dagger$ is the adjoint of $A$.  These are normalised order automorphisms when $A=U$ is unitary.  The full group of order automorphisms is generated by its identity component and a transpose (with respect to an arbitrary basis), but only the identity component is allowed in the GPT of quantum theory because the transpose is not completely positive. The completely positive property becomes visible once we consider GPTs with compound systems (i.e.~tensor products).
\end{example}

A normalised order isomorphism $V \rightarrow W$ necessarily sends pure states of $(V_+)_u$ to pure states of $(W_+)_u$.
In sufficiently well-behaved GPTs, discrete instruments correspond to indexed sets of subnormalised positive maps $\Phi_i$ that sum to a normalised one.  
Such maps $\Phi_i$ can be interpreted as describing the change of state associated with a measurement result whose probability is $\inn{u,\Phi_i(\omega)}$.    
Up to rescaling, the (unnormalised) order-isomorphisms are those GPT transformations $\Phi:V \rightarrow W$ that are \emph{probabilistically reversible}.
These are transformations $\Phi:V \rightarrow W$ for which there is a sub-normalised positive 
map $\Phi^\sharp$ such that for all $\omega \in \Omega$, 
$\Phi^\sharp \circ \Phi(\omega) = p_\omega \omega$, where  $p_\omega > 0$.  
Although the definition allows $p_\omega$ to be $\omega$-dependent, it cannot be since
$\Phi^\sharp \circ \Phi$ must preserve convex combinations.  Hence $\Phi^\sharp \circ \Phi = p \: \id$,
with $p>0$, so that if $V=W$, $\Phi^\sharp = p \: \Phi^{-1}$ and $\Phi$ is an order-isomorphism. 

The behaviour of automorphisms with respect to the direct summand structure
of reducible cones is relatively elementary; see e.g. 
Proposition 2 of Ch. I, $\S$1 of \cite{VinbergStructure} and the discussion preceding it:

\begin{proposition}\label{prop: automorphisms of reducible cone}
Let $V^+ = \oplus_{i=1}^m V^+_i$ be a direct sum decomposition of the
cone $V^+$ into irreducible cones.  
Then $\aut{(V_+)}$ is identical to the set of maps
\begin{equation}
(\oplus_i \Phi_{i},\sigma_i)
\end{equation}
where $\sigma_i$ is a permutation of $\{1,...,m\}$ and 
each $\Phi_i: V_i \rightarrow V_\sigma(i)$ is an isomorphism.
\end{proposition}



%

\subsection{Self-duality, homogeneity, and pure transitivity}
Three properties that a GPT state space can have especially interest us: self-duality, homogeneity, and pure transitivity.

\begin{definition}\label{def: self-duality}
    Let $V$ be an ordered vector space. We say $V$ is \emph{self-dual} when the space $V$ can be equipped with an inner product $\inner{\cdot,\cdot}$ satisfying $a\geq 0$ iff $\inner{a,b} \geq 0$ for all $b\geq 0$.
\end{definition}
Such an inner product  identifies the ordered vector space $V^*$ with $V$ and makes them order-isomorphic: \added{By nondegeneracy of the inner product, the map $V\to V^*$ given by $a\mapsto \inn{a,\cdot}$ is a linear bijection, and we see that 
if $a\geq 0$ then this maps $a$ to a positive element of $V^*$, and vice versa (by symmetry of the inner product) the pre-image of any positive $e\in V^*$ is a positive element of $V$.}
\begin{example}
    In the GPT of quantum theory, each system $V=M_n(\mathbb{C})_{\text{sa}}$ is self-dual: we indeed have $\omega\geq 0$ if and only if $\inner{e,\omega} = \tr(e\omega) \geq 0$ for all $e\geq 0$. This is easy enough to see: 
for any $\omega \in V$ 
pick an eigendecomposition $\omega=\sum_j \lambda_j\ketbra{\psi_j}{\psi_j}$. Setting $e=\ketbra{\psi_j}{\psi_j}$ we then have $\inner{e,\omega} = \lambda_j$;  the only way this is always positive is if $\lambda_j\geq 0$ for all $j$. The self-duality of quantum systems is the reason that we have been identifying the effects $e\in V^*$ with positive elements $e\in V$ in all the previous examples.
\end{example}
\begin{remark}
    Note that for self-duality it is not enough just to ask for $V$ and $V^*$ to be order-isomorphic.  We call that property \emph{weak self-duality}. It is strictly weaker than self-duality: $V \oplus V^*$ is always weakly self-dual but is only self-dual when $V$ is. There also exist irreducible examples~\cite{BBLW08} and even homogeneous ones~\cite{IshiNomuraArbitraryRank}.
    While there is an operational characterisation of weak self-duality related to the existence of conclusive  teleportation protocols in certain types of GPT composites~\cite{BBLW08}, we do not know of any compelling operational motivation for `strong' self-duality on its own.
\end{remark}

\begin{definition}
    The cone $V_+$ of an ordered vector space $V$ is \emph{homogeneous} when its automorphism group acts transitively on its interior.
    That is, for all $a,b\in V_+$ in the interior of the positive cone there exists an order automorphism $\Phi:V\to V$ such that $\Phi(a) = b$. 
\end{definition}
\added{Instead of stating that the cone $V_+$ is homogenous, we will sometimes just say that the space $V$ is homogeneous, meaning that its cone of positive elements is homogeneous.}
\begin{example}
    In the GPT of quantum theory, each system $V=M_n(\mathbb{C})_{\text{sa}}$ is homogeneous. Indeed, recall that in finite-dimensional quantum theory, an unnormalised state is in the interior of the positive cone iff it is full-rank, and hence invertible. Then, for two such operators $\rho$ and $\sigma$ we can define an order isomorphism $\Phi: 
    a \mapsto \sqrt{\sigma} \sqrt{\rho}^{-1} a \sqrt{\rho}^{-1} \sqrt{\sigma}$, so that $\Phi(\rho) = \sigma$.
\end{example}

In the context of GPTs, homogeneity is equivalent to strictly positive states being probabilistically reversibly transformable into one another.
Homogeneity can also be operationally motivated as being necessary to uniformly \emph{steer} states of a system~\cite{BarnumSteer}, as we will explore later.
Intuitively, if a cone is homogeneous, it is `maximally symmetric', since on an order-theoretic level, every interior element is equivalent to any other.
Note that the condition on the states being interior is necessary in the definition of homogeneity: there is no order automorphism that maps a boundary state to an interior state (in the context of quantum theory: order automorphisms always preserve the rank of an operator).

A cone is homogeneous if and only if its dual is~\cite{VinbergHomogeneousTheory}. 
From this and the fact that the duals of order isomorphisms are precisely the order isomorphisms of the dual space it follows that in a homogenous cone any two functionals $u$ and $u'$ in the interior of the dual cone 
can be mapped to each other by order isomorphisms, so the cone has a unique base up to affine isomorphism. \added{In particular, while choosing a different base will give a different set 
of pure states, the sets of (pure) states will be related to each other by an order isomorphism, and hence are equivalent in a strong sense (they are affinely isomorphic). We will hence just assume that for a homogeneous cone some base has been fixed, and hence that we have a notion of (pure) state.}

\begin{remark}
    \added{We saw for the quantum GPT that there was a difference between positive maps and completely positive maps. In light of the exclusion of the transpose from the ``allowed'' transformations in the quantum GPT, it is hence important to note that for the group of automorphisms of a regular cone to act transitively on its interior, it is necessary and sufficient that its identity component act transitively (see e.g. \cite{VinbergHomogeneousTheory}), so that homogeneity continues to be operationally relevant under this restriction.}
\end{remark}

\begin{remark}
    \added{
    Parts of the mathematical literature, and in particular the work of Vinberg that we will rely on
    \cite{VinbergHomogeneous, VinbergHomogeneousTheory, VinbergStructure}, consider a notion of cone that differs from ours in that it is a finite-dimensional \emph{open cone}:  an \emph{open} set in a finite-dimensional real vector space, closed under addition and under multiplication by \emph{strictly positive} real scalars.  It is homogeneous if its affine automorphism group acts transitively on it.   In our finite-dimensional setting, it is relatively elementary that 
    open cones are precisely the interiors of the closed cones we have been considering (which are therefore precisely the closure of open cones).  Since we will define homogeneity of closed cones as transitive action of the automorphism group on the interior, open cones are homogeneous if and only if their closures are.  The bijection between open and closed cones, and the fact that their automorphism groups are identical, allows us to apply most results of e.g.~\cite{VinbergHomogeneous, VinbergHomogeneousTheory, VinbergStructure} and other works concerning open cones, to their closures.
    }
\end{remark}

The third property, \emph{pure transitivity}, is actually a property of a compact convex set instead of the full cone.
\begin{definition}
    Let $\Omega\subseteq V$ be a base of an ordered vector space $V$. We say that $\Omega$ has \emph{pure transitivity} when for any two pure $\omega_1,\omega_2\in\Omega$ there is a normalised order-automorphism $\Phi$ of $V$ mapping the first to the second: $\Phi(\omega_1)=\omega_2$.
\end{definition}
\begin{remark}
    The property of pure transitivity is independent of how $\Omega$ is affinely embedded into $V$, and since for a homogeneous cone there is a unique base up to affine isomorphism, we will simply say that $V$ has pure transitivity, instead of specifying it as a property of a chosen base.
\end{remark}

Pure transitivity has the simple operational interpretation that the reversible dynamical possibilities of a theory should be rich enough that every pair of pure states of the same system should be reversibly transformable into one another.  If one were to take  the perspective that the pure states are the `true' states of the theory (which is a common viewpoint, e.g.~in most versions of the ``many worlds'' interpretation, and also in many approaches to the quantum theory of open systems), and the reversible dynamics the `true' dynamics, pure transitivity states that we can in fact prepare a system in any state of our choosing using the dynamics of the theory. 
\added{From this point of view,  one can imagine what it would mean for a system to \emph{not} have this property. Then when you are given the system in a certain pure state, it might not be possible to reach your desired pure state using the `true' dynamics of the theory. This almost necessitates calling systems given in such incompatible pure states distinct types of systems, which cannot be transformed into each other.}
  
Pure transitivity, and its slightly stronger version, \emph{continuous} pure transitivity defined below, are assumptions often used 
(cf.~e.g.~\cite{Hardy2001a,MasanesAxiom,MasanesEtAlEntanglementAndBlochBall}) in the theory of \emph{generalised probabilistic theories} (GPTs)~\cite{BarrettGenProb}.
The assumption also follows from \emph{essential uniqueness of purification} in combination with \emph{pure conditioning} (that pure measurements preserve pure states)~\cite{GiulioDer11}, two other assumptions often made in the field of reconstructions of quantum theory.

\added{
Pure transitivity is a requirement on the possible dynamics of a system, yet it only requires discrete transitions between pure states (for instance, a classical GPT in the form of a simplex has pure transitivity even though its set of pure states is discrete). However, in our current physical theories of the real world, dynamics is modelled as happening gradually over the course of a continuous time evolution. We can define a stronger notion of pure transitivity that expresses this property:
}
\begin{definition}\label{def:cont-pure-trans}
    We say $\Omega$ has \emph{continuous} pure transitivity when for any two pure states $\omega_0$ and $\omega_1$ there is a continuous path $\omega_t$, $t\in [0,1]$ in state space, such that $\Phi_t(\omega_0) = \omega_t$ for some normalised order isomorphism $\Phi_t$.
\end{definition}

\begin{example}
    In quantum theory, the normalised order-automorphisms are unitary conjugations, $\rho \mapsto U \rho U^{-1}$. The pure states are of the form $\ketbra{\psi}{\psi}$. Hence, to map with a normalised order-automorphism from $\ketbra{\psi}{\psi}$ to $\ketbra{\phi}{\phi}$ we simply pick any unitary $U$ that satisfies $U\ket{\psi} = \ket{\phi}$. Quantum theory additionally has continuous pure transitivity, since any unitary $U$ can be written as $e^{it H}$ for some self-adjoint $H$ and real $t$, so that the evolution from one pure state to another can be done continuously by varying $t$ from 0 to 1.
\end{example}

Note that in this example of quantum theory the maps themselves evolve continuously, giving a continuous path for every input state, which is a stronger property than we require in Definition~\ref{def:cont-pure-trans}. This stronger property would make sense to require, but we will not need it.  

\subsection{Jordan algebras}

Finally, before we can state and prove our main result, we need to introduce Euclidean Jordan algebras.
\begin{definition}
    A Jordan algebra $(V,*,1)$ is a real vector space $V$ equipped with a commutative bilinear unital product $*$ that satisfies the identity $(a*a)*(b*a) = ((a*a)*b)*a$. It is \emph{Euclidean} if it is equipped with an inner product satisfying $\inn{a*b,c} = \inn{b,a*c}$.  
\end{definition}

We will only consider \emph{Euclidean} Jordan algebras (EJAs) in this paper.
The motivating example takes $V$ to be the space of self-adjoint matrices, and sets $A*B := \frac12 (AB+BA)$. Every Jordan algebra is a product of simple algebras (hence a direct sum of these as ordered vector spaces), which have been fully classified~\cite{JNW}. 
These are either real, complex or quaternionic spaces of self-adjoint matrices, equipped with the
symmetrised matrix product as Jordan product, or a type of algebra known as a spin factor (the precise definition of which won't be relevant for us, but which correspond to systems whose normalised state spaces are the finite-dimensional balls), or the so-called Albert or exceptional algebra, consisting of of the $3\times 3$ octonionic self-adjoint matrices.  Each of these types of algebras is embeddable into an algebra of complex self-adjoint matrices (i.e.~quantum systems), except for the Albert algebra.  
An EJA is ordered by setting $a\geq 0$ iff $a=b*b$ for some $b$. 
As an ordered vector space it is homogeneous and self-dual. The Koecher-Vinberg theorem is the converse.
\begin{theorem}[K{}oecher-Vinberg]
	Let $V$ be a self-dual and homogeneous ordered vector space. Then it is order-isomorphic to a Euclidean Jordan algebra.
\end{theorem}

While the definition of a Euclidean Jordan algebra might seem esoteric, the classification of the possibilities for simple EJAs shows that, with the exception of spin factors, we can view them as algebras of self-adjoint matrices over the classical division algebras. Alternatively, with the exception of the Albert algebra $M_3(\mathbb{O})$, we can view them as subalgebras of quantum systems.  
\added{
Moreover, they share many geometrically, informationally, and physically significant properties 
with standard quantum systems, including those we have discussed above, but also other 
important properties such as having spectral decompositions.}  For these reasons EJAs are often considered a `slightly' more general class of systems that include real-valued, complex-valued, and quaternion-valued quantum systems as special types.
This seems to be reflected in reconstructions of quantum theory:  once you have shown your systems to correspond to EJAs, 
\added{the reconstruction is easily completed using any of several known characterizations of complex matrix algebras as a subset of the EJAs, as we will explain in section \ref{subsec: reconstructing quantum}.}


\section{Main result}

We are now ready to state and prove our main result.  
\begin{theorem}\label{theorem: main}
	Let $V$ be a homogeneous ordered vector space satisfying pure transitivity. Then $V$ is self-dual and hence can be given the structure of a Euclidean Jordan algebra.
\end{theorem}

Our proof is an application of Vinberg's theory of homogeneous cones~\cite{VinbergHomogeneousTheory} and their automorphism groups~\cite{VinbergStructure}, using also a result from \cite{TruongTuncel}. 
We sketch the proof here and present the details, including the necessary additional background later in this section.  
\begin{proof}[Proof sketch]
    In \cite[Ch.~III]{VinbergHomogeneousTheory}, Vinberg showed that a homogeneous $V$ can be embedded
    as the Hermitian subspace in a type of real algebra with involution 
    and inner product $(\cdot,\cdot)$, called a \emph{$T$-algebra}. 
 In \cite{VinbergStructure}  he showed additioanlly that for a homogeneous $V$ there is a subspace $V^c$ called the \emph{kernel} 
    such that the cone $V^c_+ := V^c \intersect V_+$ is homogeneous and self-dual with respect to the restriction of the inner product to $V^c$.
    He also showed that the kernel is invariant under order automorphisms
    of $V$ that preserve the $T$-algebra unit $u$ \cite[Thm.~3]{VinbergStructure}.  
    We show that these order-isomorphisms are also the normalised order-automorphisms,
    i.e.~that they preserve the base $\Omega := \{x \in V: (u,x)=1\}$.
    From \cite[Thm.~2]{TruongTuncel} one straightforwardly derives that there is at least one pure state $\omega_c$ of $V^c$ that is also a pure state of $V$ \added{(for this particular chosen base)}.  By pure transitivity of $V$, there is then for every pure state $\omega$ of $V$ a normalised order isomorphism $\Phi$ such that $\Phi(\omega_c) = \omega$. But since $V^c$ is invariant under normalised order isomorphisms, we also have $\omega \in V^c$. Thus $V$ and $V^c$ have the same pure states, and hence must be equal, so that $V=V^c$ is self-dual.
\end{proof}

\added{
In the rest of this section we give a detailed proof of the main theorem.  
The remainder of the paper can be understood without these details, however, so that a reader may skip to Section~\ref{sec:applications}.} 

\added{
\subsection{The Siegel Cone construction}
}
In order to prove our main result we need to refer to the iterative \emph{Siegel cone} construction of homogeneous cones. This appears under different guises in the literature, and so we will explain the necessary connections here, loosely following the presentation in \cite{TruongTuncel} (or 
\cite{GulerTuncelCharacterization}).

\begin{definition}
    Let $V$ be an ordered vector space with homogeneous positive cone $C$. A \emph{$V$-bilinear symmetric form} $B:\R^p\times \R^p \to V$ is a symmetric bilinear map satisfying $B(u,u)\in C$ for all $u\in \R^p$ and $B(u,u)=0$ iff $u=0$. We call $B$ \emph{homogeneous} iff there exists a subset $G$ of the order isomorphisms of $V$ that is transitive on $G$ and such that for all $\Phi\in G$ there is a linear transformation $\overline{\Phi}$ of $\R^p$ such that $\Phi B(u,v) = B(\overline{\Phi}u, \overline{\Phi}v)$ for all $u,v\in \R^p$.
\end{definition}

\begin{definition}\label{def:Siegel-cone}
    Let $V$ be a homogeneous ordered vector space with positive cone $C$ and homogeneous $V$-bilinear symmetric form $B:\R^p\times \R^p \to V$. Then the \emph{Siegel cone} $SC(V,B)$ is the positive cone of $V\times \R^p\times \R$ defined by
    \[
        (a,u,t) \in SC(V,B) \iff t\geq0 \ \text{and} \  t\cdot a - B(u,u) \in C
    \]
\end{definition}

As we shall see, the space $\R^p$ should be viewed, heuristically for now and literally once we
have the $T$-algebra formalism, as ``off-diagonal elements'' between elements of the initial 
cone $V_+$ and the newly introduced extremal ray represented as $\R$.  If they are
nonzero, i.e. $\R^p \notiso \{0\}$, this expresses the fact that the cone is not a direct sum of this ray and $V_+$, i.e. that there is ``coherence'' between these in the larger cone.    

\begin{proposition}[\cite{VinbergStructure}; cf. \cite{GulerTuncelCharacterization}, Lemma 3.1, or
\cite{TruongTuncel}, Lemma 1, for the present formulation.]
    If $V$ and $B$ satisfy the conditions above, then $SC(V,B)$ is also a homogeneous cone.
\end{proposition}

\begin{proof}
    Since it will be relevant to us, let us describe the important part of the proof here: which order isomorphisms suffice to show the cone is indeed homogeneous.  There are three types, defined respectively for any $\alpha\in\mathbb{R}_{{>}0}$, $v\in\mathbb{R}^p$ and $\Phi\in G$:
    \begin{equation*}
        T_1(a,u,t) \ := \ (a,\sqrt{\alpha}u, \sqrt{\alpha}t), \ \alpha\in\mathbb{R}_{{>}0}
    \end{equation*}
    \begin{equation*}
        T_2(a,u,t) \ := \ (a+2B(u,v) + tB(v,v), u+tv, t), \ v\in\mathbb{R}^p
    \end{equation*}
    \begin{equation*}
        T_3(a,u,t) \ := \ (\Phi(a), \overline{\Phi}u, t), \ \Phi\in G 
    \end{equation*}
    Now suppose $(a,u,t)\geq 0$. First, if $t=0$, then $a=0$ and $u=0$ and so it is not in the interior of the positive cone, so we may assume $t>0$.
    Using $T_1$ we can then always rescale the element so that $t=1$. Then, taking $T_2$ with $v=-u$ we can set the $u$ component to zero. Finally, using $T_3$ we can make the $a$ component anything in $V$ using the homogeneity of $V$ itself. Hence, any interior positive element can be 
transformed 
to any element $(a,0,1)$, giving homogeneity.
\end{proof}

\begin{example}
    Let $V=\R$ with positive cone $\R_{\geq 0}$ the nonnegative real numbers. Equip $\R^p$ with any inner product $B = \inn{\cdot, \cdot}$. Then we get the Siegel cone $(a,v,t) \geq 0 \iff t\geq 0 \ \text{and} \ ta - \inn{v,v} \geq 0$. We can represent this in a matrix form:
    \[
        \begin{pmatrix}
            t & v^T \\
            v & a
        \end{pmatrix}
    \]
    The conditions for positivity are then exactly the conditions for this matrix to be positive-definite in the usual sense. Indeed, if $p=1$ we get precisely the cone of $2\times 2$ positive symmetric matrices. More generally, if we take $V=\text{Sym}(n)$ to be the space of symmetric real $n\times n$ matrices with the positive cone consisting of positive semi-definite matrices and we set $B:\R^n\times \R^n \to \text{Sym}(n)$ to be $B(u,v) = \frac12 (uv^T + vu^T)$, then the Siegel cone $SC(V,B)$ is isomorphic to the $(n+1)\times (n+1)$ positive semi-definite symmetric matrices. In this case we could take $G$ to consist of the automorphisms $\Phi_A(M) = AMA^T$ for any invertible matrix $A$. Then $\overline{\Phi_A}u = Au$.
\end{example}

This example illustrates the fact that we can use the Siegel cone construction to construct larger and larger homogeneous cones.  Every homogeneous cone can be built by iterating the Siegel cone construction, starting from $V=\R$.   The \emph{rank} of a homogeneous cone is the 
number of iterations, which is invariant under different ways of iteratively constructing the cone and 
under order isomorphisms.  While this is often noted in the literature, an explicit proof in terms of this language is hard to find.  It is however implicit in Vinberg's characterisation of homogeneous cones in terms of \emph{T-algebras} in \cite{VinbergHomogeneousTheory}
(see the discussion after Proposition \ref{prop: closed homogeneous cone in T algebra} below), and also 
follows from Th\'{e}or\`{e}me III in \cite{KoszulOuverts}.  

\added{

\subsection{T-algebras}

}

A T-algebra is a type of not-necessarily-associative matrix algebra introduced by Vinberg in order to give an algebraic description of homogeneous cones. 
\added{A more easily accessible reference for essentially the same definition is \cite{ChuaRelating}, and a 
slightly reformulated but equivalent version is in \cite{ChuaPrimalDual}.}  Recall that an 
\emph{antiautomorphism} of an algebra is a nondegenerate linear map $x \mapsto x^\star$ such that $(xy)^\star = y^\star x^\star$.  We define an \emph{involution} of an algebra to be an antiautomorphism $x \mapsto x^*$ such that $(x^*)^* = x$.   

\begin{definition}
    A \emph{matrix algebra with involution}, $T=\bigoplus_{i,j=1}^r T_{ij}$, is a bigraded real vector space, so that we can view its elements as matrices of elements $t_{ij}\in T_{ij}$, equipped with a bilinear product 
satisfying the following properties.
    \begin{itemize}
        \item The product of an element in $T_{ij}$ with an element of $T_{jk}$ gives an element of $T_{ik}$. The product is extended by bilinearity following the usual matrix multiplication rules.
        \item $T$ is equipped with an involution $x\mapsto x^*$ that maps elements from $T_{ij}$ to elements of $T_{ji}$. I.e.~$T_{ji} = T_{ij}^*$. 
\end{itemize}
\end{definition}

An element
$t\in T$ is called \emph{upper triangular} if $t_{ij} = 0$ for $i>j$, and 
\emph{Hermitian} if $t=t^*$.  We write $T^H$ for the subspace of Hermitian elements, 
$\cT$ for the subspace of upper triangular elements, $\cT_+$ for the subset of $\cT$ whose
diagonal is nonnegative, and $\cT_{++}$ for the subsetof $\cT$ with strictly positive diagonal.  

\begin{definition}[\cite{VinbergHomogeneousTheory}, Ch.~III,  Defn.~3, cf.~\cite{ChuaRelating}]
 A \emph{$T$-algebra} $T=\bigoplus_{i,j}^r T_{ij}$ of rank $r$ is a matrix algebra with involution
satisfying the following properties.
    \begin{enumerate}
        \item 
The diagonal algebras $T_{ii}$ are isomorphic to $\R$; call this isomorphism $\rho_i$.  We'll write $e_i$ for the element of $T_{ii}$ corresponding to $1 \in \R$ under this isomorphism (i.e. 
$e_i := \rho^{-1}(1)$).
\item  
\label{item: diagonal algebra action} 
For all $j$, $e_i$ acts as the identity on $T_{ij}$.  
\item There is a linear map, the trace $\tr: T \mapsto \R$, such that $\tr(ab) = \tr(ba)$, 
$\tr(a(bc)) = tr((ab)c)$, and $\tr(a^*a) >0$ unless $a=0$.  
\item $t(uw) = (tu)w$ when $t,u,w \in \cT$.
\item $t(uu^*) = (tu)u^*$ when $t, u \in \cT$.
\end{enumerate}
\end{definition}

From item \ref{item: diagonal algebra action} it is immediate that $\sum_i e_i$ is the unit of a $T$-algebra.  


Except for the structure of matrix algebra with involution, we will not explicitly need most of the other properties appearing in the above definition, but will need the following consequence of them, which is one of the main results of \cite{VinbergHomogeneousTheory}: 

\begin{proposition}[\cite{VinbergHomogeneousTheory}, CH. III, Proposition 1]
Let $T$ be a $T$-algebra.  The set   
 $\{tt^*: t \in \cT_{++} \}$  is an open homogeneous cone in $T^H$.  
Every open homogeneous cone is order-isomorphic to the open cone $\{tt^*: t \in \cT_{++} \}$ of a $T$-algebra.
\end{proposition}

Removing ``strictly'' and ``open'' in the above and replacing $\cT_{++}$ with $\cT_+$ gives, 
relatively straightforwardly, the analogous theorem for closed homogeneous cones.  See for instance   
\cite{ChuaPrimalDual}, where it is stated following the proof of Theorem 2 (submultiplicativity of the Euclidean norm under the $T$-algebra product). 

\begin{proposition}\label{prop: closed homogeneous cone in T algebra}
Let $T$ be a $T$-algebra.  The set  
 $\{tt^*: t \in \cT_{+}\}$ is a closed homogeneous cone in $T^H$.  
Every closed homogeneous cone is order-isomorphic to the closed cone $\{tt^*: t \in \cT_{+} \}$ of a $T$-algebra.
\end{proposition}

From the T-algebra representation of a homogeneous cone we can directly read off---in reverse order---its chain of Siegel cone constructions: given a rank $r$ T-algebra $T$, we let $V =\bigoplus_{i=2,j=2}^r T_{ij}$ be the sub T-algebra where we start the index at 2, i.e.~we slice off the first row and column.   We then pick $\R^p \cong T_{21}\oplus T_{31}\oplus \cdots\oplus T_{r1}$ (so $p = \sum_{j=2}^r n_{j1}$ to be the first column minus $T_{11}\cong \R$. The $V$-bilinear symmetric form can then be defined by $B(t,s) = \frac12 (ts^* + st^*)$ which indeed takes a value in $V$ due to the $T$-algebra matrix multiplication rules. 
We then have $T \cong SC(V, B)$.   This gives the $r$-th step in the iterative construction
(if we consider starting with $\R$ to be the first step), and applying the same argument to the subalgebra $V$ gives the $r-1$st step, etc.  We hence see that the rank of the $T$-algebra coincides with the notion of rank coming from the iterated Siegel cone construction, and that the $T$-algebra representation is essentially an explicit description of this iterated construction, beginning with $\R$ (which ends up being $T_{rr}$).

Finally, let us describe the `normalised' elements of a T-algebra.
\begin{definition}
Let $T$ be a T-algebra.  We define the \emph{canonical inner product} $(a,b) = \tr(ab)$.  
\end{definition}

Let $V$ be a homogeneous vector space for which $V\cong T^H$. Then we may consider 
the particular base $\Omega := \{x \in V_+: (u,x)=1\}$ of $V_+$, where $u$ is the unit of the $T$-algebra.
We have noted that for a homogeneous cone all bases are equivalent up to order isomorphisms, so we lose no generality by picking the base consisting of the $(u,x)=1$ elements.


\added{

\subsection{The kernel of a homogeneous cone}

}

In \cite[Ch.~I~\textsec 4]{VinbergStructure}, Vinberg showed that 
there is a subspace $V^c$ of $V$ called the \emph{kernel} (the Hermitian part of the kernel 
$T^c$ of the $T$-algebra), 
such that the cone $V^c \intersect V_+$ is homogeneous and self-dual. 
Since the kernel will be the essential component in the proof of our main theorem, we cover some of the details of its construction to make clear where it comes from.

First, we represent $V$ by its T-algebra $T=\bigoplus_{ij}^rT_{ij}$. Let $n_{ij} = \text{dim}(T_{ij})$ be the dimension of each subspace. Note that $n_{ii} = 1$ and $n_{ij} = n_{ji}$. Vinberg shows that $V$ is self-dual and irreducible if and only if all $n_{ij}$ for $i\neq j$ are equal and non-zero.  
In that case the T-algebra is associative and if the rank $r$ is greater than 3, corresponds to a matrix algebra over the real, complex or quaternionic field. If the rank is $3$, the octonionic matrices are also a possibility, and if it is $2$, also the spin factors.
The idea behind the kernel is that if not all the $n_{ij}$ are equal, we partition the indices 
into subsets where they are equal, and observe that the subalgebra of elements $t$ with all $t_{ij} = 0$ except for $i,j$ in the same subset, satisfies the requirements for being self-dual.

In more detail, Vinberg shows that we can replace $T$ by an isomorphic $T'$ that is related by an `inessential change in the gradation' (i.e.~a relabeling of the indices and corresponding subspaces), such that this $T'$ allows an equivalence relation $R$ on the indices $1,\ldots,r$ that satisfies the following two properties for all $i<j$:
\begin{equation}
    (A) \ \ i\,R\,j \implies n_{ij} \neq 0 \text{ and for all } k\neq i,j,  n_{ik} = n_{jk}
\end{equation}

\begin{equation}
    (B) \ \ i<k<j \text{ and } i\,R\,j \implies i\,R\,k \text{ and } j\,R\,k.
\end{equation}

We then build $T_c$, the kernel, where we keep only the $T_{ij}$ spaces of $T$ where $i\,R\,j$, and set all the others to zero:
\begin{equation}
    (T_c)_{ij} \ = \ 
    \begin{cases}
        T_{ij} \text{ if } i\,R\,j   \\
        0 \text{ else}. 
    \end{cases}
\end{equation}
Letting $M_1,\ldots,M_s$ denote the partition of the indices $1,\ldots,r$ given by the equivalence relation $R$, we see that $T_c$ can be seen as a direct sum of T-algebras $T^{(l)}$ for $l=1,\ldots,s$ where the indices of $T^{(l)}$ come from $M_l$. Due to property $(A)$ we have equal dimensions of the off-diagonal spaces in each $T^{(l)}$ and hence each $T^{(l)}$ gives an irreducible self-dual cone. As $T_c$ is a direct sum of these self-dual cones, we see that $T_c$ itself is also a self-dual cone.

We collect in the following proposition several properties of the kernel.
\begin{proposition}\label{prop:kernel-properties}
    Let $V=T^H$ and let $T_c$ be $T$'s kernel as defined by Vinberg~\cite[Ch.~I~\textsec 4]{VinbergStructure}. Then the following holds.
    \begin{itemize}
        \item $V^c = V \intersect T_c \ (= V \intersect T_c^H)$ has homogeneous cone $V^c_+ = V_+\intersect T_c^H$, which has T-algebra $T_c$
        \item  $V^c$ contains all the diagonal spaces $T_{11},\ldots,T_{rr}$.  So $V^c$ is non-zero if $V$ is.
        \item $V^c_+$ is a homogeneous, self-dual cone.
        \item Given a base $\Omega$ of $V$, its restriction $\Omega_c := V^c \intersect \Omega$ is a base of $V^c_+$.
    \end{itemize}
\end{proposition}

There is one final property that is essential to our proof.

\begin{theorem}[{\cite[Theorem~3]{VinbergStructure}}(p. 88)]\label{thm:kernel-preserve-unit}
    Let $V_+$ be a homogeneous cone with T-algebra $T$ and T-algebra unit $u$. Then its kernel $V^c_+$ is invariant under order automorphisms of $V$ preserving $u$, i.e.~if $\Phi(u) = u$, then $\Phi(V^c_+) = V^c_+$.
\end{theorem}



As a final step, we need to relate unit-preserving isomorphisms to base-preserving isomorphisms.

\begin{proposition}\label{prop:base-preserving is unit-preserving}
    The order isomorphisms of $V\cong T^H$ that preserve the base $\Omega := \{x \in V_+: (u,x)=1\}$ are precisely those that preserve the unit: $\Phi(u) = u$.
\end{proposition}
\begin{proof}
    In \cite[Ch.~I~\textsec 2-3]{VinbergHomogeneousTheory}, Vinberg constructs a canonical Riemannian metric $g$ on an arbitrary open cone $V_+^\circ$ 
(which we will view as the interior of its closure $V_+$)
and shows
    it is invariant under order automorphisms $\Phi$.  That is, for any $p \in V_+^\circ$
    and $x,y$ in the tangent space at $p$, and identifying the
    tangent space with the ambient space $V$ via the Riemannian metric,
    $g_{\Phi(p)}(\Phi(x), \Phi(y)) = g_p(x,y)$. 
In the case of a homogeneous cone, he shows that Riemannian metric may be chosen 
so that at the $T$-algebra unit $u$ the metric $g_u$ is equal to 
the canonical inner product of the $T$-algebra $T$ 
    \cite[II, Eq. (10) and III, Eq. (34)]{VinbergHomogeneousTheory};
    see e.g.~\cite{Tsuji} for a more compact presentation.  That is: $(a,b) = g_u(a,b)$.
    For $\Phi$ such that $\Phi(u) = u$, 
    the metric's invariance at $u$ states that $(\Phi(x), \Phi(y)) = g_u(\Phi(x), \Phi(y)) = 
    g_u(x,y) = (x,y)$, i.e the canonical inner product is invariant 
    under the group of unit-preserving automorphisms. 
    Hence, the dual of such a transformation with respect to this inner
    product is its inverse, so the group of unit-preserving transformations
    is identical, as a group of linear operators on V, to the
    group of its duals. But then the order automorphisms satisfying $\Phi(u) = u$ are precisely those that preserve the base, i.e. those such that $\Phi(\Omega) = \Omega$,  because $(u,\Phi(x)) = (\Phi^*(u),x) = (u,x) = 1$.
\end{proof}

Combining this with Theorem~\ref{thm:kernel-preserve-unit}, we see that 
the kernel $V^c$ is invariant under the normalised, i.e.~base-preserving, automorphisms.

\added{

\subsection{Pure states in homogeneous cones}
}

Our assumption on pure transitivity is a condition on pure states. 
But Vinberg was concerned with \emph{open} homogeneous cones, whose closures are what we have been calling homogeneous cones. 
Since pure states lie on the boundary of the positive cone, Vinberg's results do not include explicit claims about pure states.  
However, Vinberg's construction was adapted by Truong and Tun\c{c}el~\cite{TruongTuncel} to describe the \emph{extremal rays} of closed homogeneous cones, and hence the pure states of 
a chosen base, 
which are in bijection with the extremal rays containing them.

To explain this further, we begin by reviewing some general facts about the geometry of 
(closed, pointed, but not necessarily homogeneous)  cones.

\begin{definition}\label{def:face}
    Let $V$ be an ordered vector space with positive cone $V_+$.
    A \emph{face} $F$ of the cone is a subset $F \subseteq V_+$ such that for each $x \in F$, 
    the only way to express $x$ as a nonnegative sum of elements in $V_+$ is if
    those elements are also in $F$.  
    That is, if $\sum_j \lambda_j v_j \in F$ with all $\lambda_j>0$, then all $v_j\in F$.
    The \emph{face generated} by $x \in V_+$, written $\face{(x)}$, 
    is the smallest face of $V_+$ containing $x$.
\end{definition}
Note that if $x$ is in the \emph{relative interior} of a face $F$ ($F$'s interior with respect to its linear span in $V$) then $\face(x) = F$.

We can also define faces not just for cones, but for compact convex sets.
\begin{definition}
    Let $\Omega$ be a compact convex set. A subset $F\subseteq \Omega$ is a \emph{face} when $\sum_j \lambda_j \omega_j \in F$ with $\lambda_j>0$ and $\sum_j \lambda_j=1$ implies that $\omega_j\in F$ for all $j$.
\end{definition}
Note that the extremal points of $\Omega$ are faces containing just a single element.
The faces of a base $\Omega$ for a cone $V_+$ are in bijection with 
the non-null faces of $V_+$ via $F\subset \Omega \leftrightarrow \R_+ F \subset V_+$.

An \emph{extremal ray} of $V_+$ is a one-dimensional face.
Equivalently, for any fixed base $\Omega$
of  $V_+$, each extremal ray of $V_+$ is the set $\R_+\omega$ of nonnegative multiples
of some pure state $\omega$ of $\Omega$, or the null ray $\{0\}$.  
In particular, we have the following:
\begin{proposition}\label{prop: extremal rays generate}
The non-null extremal rays are in bijection with the pure states of $\Omega$, and $V_+$ is 
generated by its extremal rays (indeed by the pure states of any fixed base) 
via nonnegative linear combinations. 
\end{proposition}

In \cite{TruongTuncel} the authors give a classification of the pure states (in terms of extremal rays) of a homogeneous cone by induction using the Siegel cone construction.  By Proposition 
\ref{prop: extremal rays generate}, given a choice of base this also classifies the pure states.  For any cone $V_+$, 
we write $\partial_e V_+$ for its set of extremal rays. 

\begin{theorem}[{Paraphrased from \cite[Theorem 2]{TruongTuncel}}]
    Let $V$ and $B:\R^p\times \R^p \to V$ satisfy the 
conditions of Definition~\ref{def:Siegel-cone}. 
and let $SC(V,B)$ be the resulting Siegel cone.  Then
    \[
    \partial_e SC(V,B)) = \{(a,0,0)\,|\, a\in \partial_e(V_+)\} \cup \{\R_+(B(u,u),u,1) \,|\, u\in \R^p \}
    \]
    That is, the extremal rays are those of the original cone, plus those coming from the vectors in 
$\R^p$.
\end{theorem}



\begin{corollary}
    The element $(0,0,1)$ of the Siegel cone $(V,B)$ generates an extremal ray, hence
is a unique positive multiple of a pure state.
\end{corollary}
\begin{proof}
    Just set $u=0$ in the statement above.
\end{proof}
Note that in the T-algebra representation of the iterated Siegel cone construction, the element $(0,0,1)$ corresponds precisely to the matrix with a $1$ in the top-left corner, and zeroes everywhere else. Since the kernel $V^c$ always contains the full diagonal of the T-algebra, it also contains this element $(0,0,1)$. We'll capture the relevant fact here as a corollary for later.
\begin{corollary}\label{cor:pure-state-in-kernel}
    Let $V$ be any homogeneous cone. Then there exists a pure state of $V$ that lies in its kernel $V^c$.
\end{corollary}

\subsection{Proof of main theorem}\label{app:proof of main theorem}

\added{
We now have all the ingredients to prove the main statement: that a homogeneous cone is self-dual if it satisfies pure transitivity.

\begin{proof}[Proof of Theorem \ref{theorem: main}]
By Corollary~\ref{cor:pure-state-in-kernel} there is a pure state $\omega_c$ of $V$ that 
lies in $V^c$. Now let $\omega$ be any other pure state of $V$. By pure transitivity, there exists a normalised, i.e.~base-preserving order isomorphism $\Phi$ such that $\Phi(\omega_c) = \omega$. Then by Proposition~\ref{prop:base-preserving is unit-preserving}, $\Phi$ must also preserve the T-algebra unit of $V$. Hence, by Theorem~\ref{thm:kernel-preserve-unit}, $\Phi$ sends elements of the kernel $V^c$ to the kernel, so that $\omega = \Phi(\omega_c)\in V^c$. But $\omega$ was an arbitrary pure state of $V$, so all pure states of $V$ lie in $V^c$. But then the entire base of states 
$\Omega$ lies in $V^c$, which by the last point of Proposition~\ref{prop:kernel-properties} means that the bases of $V_+$ and $V_+^c$ coincide, and hence $V_+= V_+^c$, 
which is self-dual.
\end{proof}
}

\section{Applications of the main result}\label{sec:applications}

\added{
The fact that homogeneity and pure transitivity imply self-duality and hence that the space is a Jordan algebra has a number of interesting corollaries and consequences that we will explore in this section.
}

First, note that not all Jordan-algebraic state-spaces have pure transitivity: we will show 
that an EJA with pure transitivity must be a direct sum of $n$ isomorphic simple Jordan algebras, which in GPT terms corresponds to a (locally tomographic) composite of an $n$-state classical system with a simple Jordan system.


\begin{corollary}
    Let $V$ be a homogeneous ordered vector space.  It satisfies pure transitivity iff it is order-isomorphic to a direct sum of identical simple Euclidean Jordan algebras.
\end{corollary}

\begin{proof}
\added{
    Recall that every Euclidean Jordan algebra is isomorphic to a direct sum of simple Euclidean Jordan algebras; we must show that pure transitivity implies that these summands are isomorphic.  Any unital order isomorphism on an EJA is a Jordan-algebra isomorphism, and such a map must
 preserve the direct sum structure.  I.e.~if $A=\bigoplus_{j\in J} A_j$, the 
Jordan isomorphisms 
$\Phi: A \mapsto A$ are precisely the maps  
\begin{equation}
\Phi = \sum_{j \in J} \Phi_j, 
\end{equation}
 where $\Phi_j: A_j \rightarrow A_{\sigma j}$ is a Jordan isomorphism, and 
$\sigma$ is a permutation.  
Thus the permutations $\sigma$  that can appear here are those that preserve isomorphism classes of summands.  

Every summand contains at least one pure state, and for every pair of summands $A_i$ and $A_j$ and any pair of pure states  $\omega_i \in A_i$ and $\omega_j \in A_j$, there exists a base-preserving, hence (given our choice of base)  unital by
Proposition  \ref{prop:base-preserving is unit-preserving}, order-isomorphism $\Phi$ such that 
$\Phi(\omega_i) = \omega_j$.  Preservation of the direct sum structure by order-isomorphisms then implies that $\Phi(A_i) = A_j$, hence $\Phi_i(A_i) = A_j$.  Since $A_i$ and $A_j$ were arbitrary and  $\Phi_i$ is a unital order-isomorphism, 
all summands must be isomorphic to each other.

The other direction, that if $A$ is a direct sum of isomorphic simple EJAs it has pure transitivity, follows from the pure transitivity of the simple EJA factors~\cite[Prop.~IV.3.2]{FarautKoranyi} and 
the abovementioned fact that arbitrary isomorphism-class-preserving permutations of simple summands are unital order-isomorphisms.
}
\end{proof}

\added{
Reversible transitions from one pure state to another within a direct sum factor can be done in a continuous fashion, but between pure states in different direct sum factors this is 
not possible.  This is because an order isomorphism (and hence a Jordan algebra isomorphism), must map all the pure states of one factor to one other factor.  In the GPT literature, direct summands (which are factors, in the case of EJAs) are usually interpreted as `superselection sectors', between which (coherent) ``superpositions'' are not possible.
}

As a consequence of this, the stronger condition of continuous pure transitivity is 
satisfied by the simple Jordan algebras (see e.g.~Cor.~IV.2.7 and Prop.~IV.3.2 of \cite{FarautKoranyi}), and only these (as in fact in any GPT system with continuous pure transitivity, the state space must be irreducible since direct summands of $V$ by definition intersect only in $0 \in V$, and all pure states of $\Omega$ lie in some direct summand so that there is no continuous path of pure states between distinct direct summands.).

\begin{corollary}\label{cor:homogeneity-simple}
    Let $V$ be a homogeneous ordered vector space. It satisfies continuous pure transitivity iff it is order-isomorphic to a simple Euclidean Jordan algebra.     
\end{corollary}

\subsection{Composite systems and steering}\label{sec:steering}
The above results don't require any notion of composite system. By introducing composites, we can give an operational motivation for homogeneity, and give a full derivation of quantum theory, in the sense of complex matrix algebras, 
by restricting Jordan-algebraic systems to just such algebras.  

\begin{definition}\label{def: composites}
    Given two systems $A$ and $B$ in a GPT with corresponding state and effect spaces, a \emph{composite system} $AB$ corresponds to a state space $\Omega_{AB}$ and effect space $E_{AB}$, such that there is a bi-affine mapping on states $\otimes: \Omega_A\times \Omega_B\to \Omega_{AB}$ and similarly one on effects, such that $\inn{\omega_A\otimes \omega_B,e_A\otimes e_B} = \inn{\omega_A,e_A}\inn{\omega_B,e_B}$ and $u_{AB} = u_A\otimes u_B$. 
    These biaffine mappings lift to bilinear positive mappings of the corresponding ordered vector spaces $V_A\times V_B \to V_{AB}$ and $V_A^* \times V_B^* \to V_{AB}^*$, which in turn lift to linear injections $V_A \otimes V_B \rightarrow V_{AB}$, $V_A^* \otimes V_B^* \rightarrow V_{AB}^*$.   
    In this setting we can define the \emph{marginal} of a composite state $\omega_{AB}$ as the state $\omega_A$ defined by the probabilities $\inn{\omega_A,e_A} := \inn{\omega_{AB},e_A\otimes u_B}$ for any effect $e_A$ on $A$.
\end{definition}

This definition embodies the operational properties of \emph{no-signalling} from $A$ to $B$ by operations done solely on $A$, and vice versa, as well as the existence 
of all product states and product effects, operationally defined as giving product probabilities when evaluated on each other.  

\begin{example}
    In the GPT of quantum theory, we make composite systems by taking the usual tensor product of the complex matrix algebras. A marginal state corresponds to taking the partial trace of a bipartite state.
\end{example}

\begin{remark}
While we have mainly been concerned with single GPT systems, we use the term generalised 
probabilistic \emph{theory} to refer to a collection of such systems that can be combined 
as in Definition \ref{def: composites} to make compound systems. In particular, if the GPT contains a system $A$, then it must also contain the composite systems $AA$ and $AAA$. 
More formally we can use the language of category theory 
to consider a GPT to be a symmetric monoidal category with additional structure corresponding to its probabilistic interpretation.
\end{remark}

We need the notion of composite systems and marginal states to talk about another property that quantum theory satisfies: steering.

\begin{definition}
    We say a bipartite state $\omega_{AB}$ of a composite $AB$ \emph{steers} its marginal state $\omega_A$ on system $A$ when for any probabilistic ensemble $\omega_A = \sum_i p_i \omega_i$ of $\omega_A$, there exists a measurement ${e_1,\ldots, e_k}$ on $B$ such that observing $e_i$ on $B$ when the full system is in the state $\omega_{AB}$, results in the state $\omega_i$ on $A$.  That is, for any effect $a$ on $A$ we have $ \inn{\omega_{AB}, a\otimes e_i} = p_i\inn{\omega_i,a}$.
\end{definition}

Schr\"{o}dinger thought it was a remarkable peculiarity of quantum theory, as compared to classical probability theory, that this `remote steering' can be done in quantum theory even though the ensembles of $\omega_A$ can consist of different sets of \emph{pure} states \cite{Schrodinger36}.  
So we can `remotely steer' the state of $A$ to effectively be drawn from our choice of one of these distinct sets of states, which may even have no state in common, by making a measurement on a system $B$ that is entangled with $A$.

\begin{definition}
    We say a system $A$ \emph{steers a system} $B$ if for every state of $B$ 
there is a state of $AB$ that steers it.  
A theory has \emph{universal uniform steering} if for every 
system $B$ in the theory, there is a system $A$ that steers it.  It has \emph{universal self-steering} if every system can be steered by a copy of itself (i.e., a system isomorphic to it).

\end{definition}

We note that both quantum theory and classical theory have universal self-steering. 

In the Appendix we prove the following theorem, which generalizes, and corrects an 
omission in the statement of, Proposition 5.6 in \cite{BarnumSteer}, giving an operational motivation for homogeneity as arising from steering.

\begin{theorem}\label{prop: universal steering implies homogeneity} 
In any theory that supports universal uniform steering  where the steering system $A$ is of the same dimension as the steered system  $B$, every irreducible finite-dimensional state space in the theory is homogeneous.
\end{theorem}
Note now (1) that continuous pure transitivity implies that state spaces are irreducible and (2) that the condition on the dimension follows from universal self-steering.  We can therefore combine Theorem~\ref{prop: universal steering implies homogeneity} with Corollary~\ref{cor:homogeneity-simple} to arrive at the following.
\begin{corollary}\label{cor: universal steering reconstruction}
    In any GPT with universal self-steering (or even just universal steering by a system of the same dimension) and continuous pure transitivity, the state spaces are those of simple Euclidean Jordan algebras.
\end{corollary}

\added{
    If we weaken continuous pure transitivity to just pure transitivity, then it is no longer clear that this result continues to hold. However, any other property of the GPT system that allows us to to define `subsystems' by taking a projection on the system, such as the \emph{ideal compression} property~\cite{GiulioDer11} or the existence of comprehensions~\cite{wetering2018reconstruction}, would then allow us to isolate irreducible systems which must then be simple Jordan algebras. Hence universal self-steering, pure transitivity and (suitably formalised) a property like ideal compression would still force all systems to be Euclidean Jordan algebras.
}

\subsection{Fully reconstructing quantum theory}
\label{subsec: reconstructing quantum}

\added{
We consider the  reconstruction of the structure of Jordan algebras based on homogeneity and pure transitivity our main result.    
In order to demonstrate the power of this result, in this subsection we describe several ways in which we can reconstruct quantum theory fully by making some additional assumptions.   
That is, we give several different recipes for singling out 
the complex matrix algebras (i.e.~quantum systems) from among all the Euclidean Jordan algebras.  
This is done by requiring one or two of several known operationally meaningful properties  \emph{not} shared with 
EJAs more generally.  
Along with the plethora of important properties general EJAs share with complex quantum theory, this shows that reconstructing the Jordan structure is a major step toward full 
reconstruction of the quantum formalism. 

We will consider in this section two general 
approaches: either (1) requiring local tomography,  or (2) 
requiring a correspondence between observables and generators of dynamics.

Some of the combinations of properties we use below yield simple matrix algebras, while others yield direct sums of matrix algebras.  
A direct sum of one-dimensional complex matrix algebras gives a 
simplex as its normalized state space, which can be viewed as a classical system.
Hence, not restricting to irreducible state spaces gives us classical, and classical-quantum mixed systems. We also find several different ways to only get the irreducible systems, or to get all complex matrix algebras.
}

\subsubsection{Properties of composite systems}

The first approach uses the well-known observation 
that EJAs other than complex matrix algebras don't have `nice' composites.  In particular, requiring the existence of \emph{locally tomographic} composites forces the EJAs to be complex matrix algebras, and hence standard quantum systems.

\begin{definition}[\cite{Hardy2001a}] 
    We say a composite system $AB$ is \emph{locally tomographic} (or \emph{tomographically local}) 
if the product effects $e_A\otimes e_B$ separate all the (possibly entangled) states of $\Omega_{AB}$. This is equivalent to $\dim V_{AB} = \dim V_A \dim V_B$.
    We say composites are \emph{purity preserving} when a composite of pure states is again a pure state, and similarly for pure effects.   
\end{definition}


\begin{theorem}\label{theorem: local tomography 1}
    Let $A$ be a system in a GPT where composites are locally tomographic and where every system satisfies homogeneity and continuous pure transitivity. 
    Then $V_A$ is a simple complex matrix algebra.
\end{theorem}
The proof uses standard dimension-counting arguments. We give a sketch here, with the details given in Appendix~\ref{app:local-tomography}. 
\begin{proof}[Proof sketch]
Continuous pure transitivity and homogeneity give us that $V_A$ is a simple EJA. Let's suppose that this EJA has dimension $d$ and \emph{rank} $r$ (i.e.~the size of the largest collection of orthogonal pure states is $r$).  By local tomography, the composite of $A$ with itself must then have dimension $d^2$.
By  Lemma 3 of~\cite{BBLW06Cloning} (for example) local tomography implies that a composite of pure states is again pure. 
For the purpose of studying the structure of the composite 
state space, we consider the systems with $A$ and $B$'s 
state spaces and unrestricted effects, so that 
purity preservation also holds for pure effects in such systems.  Using
this it is straightforward to show that the composite of a rank $r$ EJA with itself must result in a rank $r^2$ EJA. These conditions of a dimension $d^2$ simple EJA with rank $r^2$ turn out to uniquely pinpoint the complex matrix algebras using simple dimension-counting arguments.
\end{proof}

Stated a bit more informally, this theorem says that quantum theory is the unique GPT where the composites are locally tomographic, and all the system state spaces satisfy homogeneity and continuous pure transitivity. Using Corollary~\ref{cor: universal steering reconstruction} we can trade in the assumption of homogeneity for universal self-steering, giving us a completely operational reconstruction of quantum theory:
\begin{theorem}\label{thm:operational}
    Any system in a GPT that satisfies local tomography, universal self-steering and continuous pure transitivity must have a state-space isomorphic to that of a quantum system.
\end{theorem}

We can trade in the assumption of local tomography for purity preservation, to characterize a slightly broader class of systems.  
This also requires the assumption of \emph{unrestricted effects}, namely that $E$ is the entire order interval $[0,u]_{V^*} := \{x \in V^*: 0 \le x \le u\}$, i.e.~that any linear functional taking probabilities as values when evaluated on any state in $\Omega$ is an effect in the GPT. 

Luckily this assumption is not too hard to satisfy in our setting, as the following proposition shows.  
\begin{proposition}
    Let $A$ be a GPT with state space $\Omega\subseteq V$ and effect space $E\subseteq [0,u]_{V^*}$ satisfying pure transitivity and homogeneity. Then $A$ has unrestricted effects (that is: $E=[0,u]_{V^*}$) iff $E$ contains at least one pure effect.
\end{proposition}
\begin{proof}
    Pure transitivity and homogeneity imply that $A$ is an EJA, for which there is a one-to-one correspondence between pure states and pure effects.  Now let $e\in E$ be the pure effect that's assumed to be part of the GPT, and let $\omega$ be its corresponding pure state. Now let $e'\in [0,u]_{V^*}$ be any other pure effect in the dual space and let $\omega'$ be its corresponding pure state. Then by pure transitivity there is a reversible $\Phi$ in the GPT satisfying $\Phi(\omega') = \omega$. But then $e\circ \Phi = e'$ so that $e'$ is also reachable in the GPT and we have $e'\in E$. 
\end{proof}

\begin{theorem}\label{thm:real-and-complex}
    Let $A$ be a system in a GPT satisfying homogeneity (or universal self-steering) and continuous pure transitivity, and suppose that the GPT has unrestricted effects and that composites preserve purity.  Then $V_A$ is either a complex or real matrix algebra.
\end{theorem}
For the proof of this statement and several other related results, we refer the reader to Appendix~\ref{app:variations}.

The assumption of \emph{continuous} pure transitivity might also seem a bit more powerful than necessary here, since it only serves to force the systems to be irreducible, and hence simple EJAs. As we show in Appendix~\ref{app:variations}, we can get a similar result just using pure transitivity by assuming the composite is `classicality preserving', meaning when we tensor certain special effects that act `classically' together we get another classical effect.

\begin{theorem}\label{thm:non-simple-algebras}
    Let $A$ be a system in a GPT satisfying homogeneity and pure transitivity, and having unrestricted effects. Suppose furthermore that composites are purity preserving 
    (on both states and effects)
    and classicality preserving.
    Then $V_A$ is a direct sum of complex or real matrix algebras.
\end{theorem}

There is a variety of physically or informationally natural properties that would suffice to rule out classical systems in these results, such as the existence of information-disturbance tradeoffs~\cite{BarrettGenProb}, the impossibility of universal cloning or broadcasting~\cite{BBLW2007,BBLW06Cloning}, the nonuniqueness of decompositions of a mixed state into pure states, or the existence of incompatible measurements \cite{PlavalaCompatible}. 
These are easily adapted to give ways of ruling out any form of classicality, i.e.~singling out the \emph{simple} complex matrix algebras and only those.

\subsubsection{State-observable correspondence}
Instead of making assumptions on composite systems, another known way to isolate the complex matrix algebras from the other EJAs is to adapt ideas from~\cite{ASBook2} about dynamical correspondences. The property of \emph{energy observability}~\cite{BMU} generalizes to the GPT setting the fact that in quantum theory there is (for irreducible systems)  a direct correspondence between the real vector space spanned by the effects (the observables) and the ways in which the system can reversibly evolve in time: for each self-adjoint $H$, we get a time evolution $e^{-iHt}$, and vice versa for each one-parameter group of unitaries $U(t)$ we have a generator $-iH$.
Moreover, this correspondence is implemented for all observables and generators at once via a linear map, namely multiplication by $-i$. This map has the property, reminiscent of Noether's theorem, that the observable corresponding to a generator is conserved by the generated dynamics.

\begin{definition}[\cite{BMU}]
    Let $V$ be the linear space of a GPT with state space $\Omega$ and normalised order-isomorphism group Aut$(\Omega)$ and let $\fk$ be the Lie algebra of Aut$(\Omega)$. We say $V$ satisfies \emph{energy observability} when we have an injective linear map $\phi: \fk \to V^*$, such that 1) the image of $\phi$ does not contain $u$ (i.e. no non-trivial dynamics map to the trivial observable), and 2) for all $k\in\fk$ the observable $\phi(k)$ is preserved by the dynamics of $k$. That is:  $\phi(k)\circ e^{tk} = \phi(k)$ for all $t$.
\end{definition}

\begin{theorem}
    Let $A$ be a GPT satisfying universal self-steering, continuous pure transitivity and where every space satisfies energy observability. Then $A$ corresponds to a complex matrix algebra.
\end{theorem}
\begin{proof}
    Universal self-steering and continuous pure transitivity show $A$ is a simple EJA.  In~\cite{BMU} it is shown that the only simple EJAs that can satisfy energy observability are the complex matrix algebras.
\end{proof}

\begin{remark}
    In the result above, the assumption of continuous pure transitivity automatically implied $\fk$ is non-trivial. This rules out the case of finite-dimensional classical theory. The closely related postulates
    of existence of a \emph{dynamical correspondence} \cite{AlfsenShultzOrientation,ASBook2} 
    or of a \emph{Connes orientation} \cite{ConnesSelfDualCones} 
    do not require a nontrivial Lie algebra, and furthermore can also be defined for infinite-dimensional
    Jordan algebras, where these properties characterize the $C^*$ or $W^*$-algebraic cases, \emph{including} classical (commutative) ones.  
    Besides the above references, more detailed discussion on energy observability can be found in~\cite{BMU} and~\cite{BaHiSpectral}; see~\cite{BaezNoether} for a somewhat different perspective on such properties.    
\end{remark}

\section{Conclusion and discussion}%
We gave an alternative to the seminal Koecher-Vinberg theorem, exchanging 
self-duality for the arguably more operationally meaningful property of 
pure transitivity.  This slightly narrows the class of Jordan algebras characterised to just the products of isomorphic ones. Using our derivation of homogeneity from steering, this gave a simple operationally meaningful derivation of Jordan state spaces.  We then reviewed several recipes that bring us to just the complex quantum systems.  For instance: quantum systems 
 are the unique GPTs that satisfy universal self-steering, continuous pure transitivity and local tomography.  Other such recipes used several other natural properties of the tensor product 
and/or the state space.    
We also considered \emph{energy observability}, reminiscent of Noether's theorem relating the generators of reversible dynamics to conserved observables 
to narrow things down to quantum theory, as in~\cite{BMU}.
Since energy observability states that the symmetries of the pure states are generated by observables they conserve, and homogeneity and pure transitivity are also statements about symmetries, this characterizes quantum theory entirely in terms of properties of the symmetries of its state and observable space.  

Homogeneity concerns symmetries of the set of unnormalised `maximally mixed' (the strictly positive, or full-rank) states,
implying that each such state is accessible, with nonzero probability, from every other such state via a probabilistically reversible transformation.
Pure transitivity instead concerns symmetries of the normalised pure states:  they are accessible with \emph{unit} probability, via deterministically reversible transformations.   
 
These conditions are requirements on the reversible transformations of our systems, making them very powerful.  So it is natural that, as we have now shown, they have strong implications for the structure of systems, bringing them into a class of Jordan-algebraic systems that is not much larger than standard quantum theory, and shares with it the ability to support many information-theoretic protocols, many of which depend on the properties of reversible transformations.   
Indeed, the Jordan-algebraic systems share many other important and physically, informationally, and computationally significant properties with standard quantum systems.  
\added{
For example, states and effects have spectral decompositions, and the outcome probabilities of any finegrained measurement on a given state majorise the probabilities of a fine-grained spectral measurement, which allows us to define the usual von Neumann entropy and a range of other entropic quantities on these systems and to extend some of their informational and statistical-mechanical applications~\cite{QPLthermo, KrummEtAlThermo, ChiribellaScandoloEntanglementAxiomatic, ChiribellaScandoloDiagonalization, ChiribellaScandoloMicrocanonical}.
}   
We can now see that all these properties can in principle be derived from the properties of homogeneity and pure transitivity.  \added{Thus our results provide another way in which the constellation of significant properties possessed by Jordan algebraic systems can be reduced to a small subset of such properties.  Each such result provides another perspective on the mathematical and conceptual essence of quantum information and the quantum framework, while also sharpening our understanding of precisely which additional properties serve to single out the canonical, complex quantum formalism.}   


\myheading{Acknowledgments}%
HB thanks the Department of Mathematical Sciences and QMATH at the University of Copenhagen for support during a guest professorship.  

\myheading{Author contributions}%
All authors contributed equally to the research and results. HB and JvdW wrote the manuscript.

\myheading{Funding}%
JvdW was supported by a Veni personal fellowship of the Nederlandse Organisatie voor Wetenschappelijk Onderzoek (NWO).  While at QMATH HB was supported in part by projects funded by the Villum Foundation.  HB and CU thank Perimeter Institute for support; work at Perimeter Institute is supported in part by the Government of Canada through Industry Canada and by the Province of Ontario through the Ministry of Research and Innovation. 

\myheading{Conflict of Interest}%
The authors declare that they have no competing interests that would affect the results presented in this paper.

\bibliographystyle{plainnat}
\bibliography{main}

@article{wetering2018reconstruction,
    author = {van de Wetering, John},
    title = {An effect-theoretic reconstruction of quantum theory},
    journal = {Compositionality},
    volume = {1},
    issue = {1},
    doi = {10.32408/compositionality-1-1},
    link = {https://compositionality-journal.org/papers/compositionality-1-1/},
    year = {2019},
}

@article{Schrodinger36,
    Author = {Schr\"odinger, E.},
    Journal = {Proceedings of the Cambridge Philosophical Society},
    Pages = {446-452},
    Volume= {32},
    Title = {Probability relations between separated systems},
    Year = {1936}}

@article{Ludwig64a,
        author={G. Ludwig},
        title={{V}ersuch einer axiomatischen {G}rundlegung der {Q}uanten {M}echanik und allgemeinerer physikalischer {T}heorien},
        journal={Z. Phys.},
        volume={181},
        pages={233},
        year={1964}}

@article{Ludwig67a,
        author={G. Ludwig},
        title={Attempt of an axiomatic foundation of quantum mechanics and more general theories: {I}{I}},
        journal={Commun. Math. Phys.},
        volume={4},
        pages={331},
        year={1967}}

@article{VinbergHomogeneousTheory,
                author={E.~B. Vinberg},
                title={The theory of convex homogeneous cones},
                journal={Trudy Moskov. Mat. Obsc.},
                volume={12},
pages = {303-358},
                note={{E}nglish translation in {\em Trans. Moscow. Math. Soc.} {\bf 12 }, 340-403 (1965)},
                year={1963}}

@article{VinbergStructure,
                author={E.~B. Vinberg},
                title={The structure of the group of automorphisms of a homogeneous convex cone},
                journal={Trudy Moskov. Mat. Obsc.},
                volume={13},
                note={{E}nglish translation in {\em Trans. Moscow. Math. Soc.} 13: 63-93 (1965)},
                year={1965}}

@article{BBLW2007,
    Author = {H. Barnum and J. Barrett and M. Leifer and A. Wilce},
    Journal = {Phys. Rev. Lett.},
    Pages = {240501},
    Title = {Generalized no-broadcasting theorem},
    Volume = {99},
    Year = {2007}}

@article{MasanesEtAlEntanglementAndBlochBall,
        author={{Ll.} Masanes and M. P. M{\"u}ller and R. Augusiak and D. P{\'e}rez-Garcia},
        title={Entanglement and the three-dimensionality of the {B}loch ball},
        journal={J. Math. Phys.},
        volume={55},
        pages={122203},
        year={2014},
        doi = {10.1063/1.4903510}
        }

@unpublished{Hardy2001a,
    author={L. Hardy},
    title={Quantum theory from five reasonable axioms},
    note={arXiv.org e-print {\tt quant-ph/0101012}},
        year={2001}}

@inproceedings{IshiNomuraArbitraryRank,
        author={H. Ishi and T. Nomura},
        title={An irreducible homogeneous non-selfdual cone of arbitrary rank linearly isomorphic to the dual cone},
        booktitle={Infinite Dimensional Harmonic Analysis IV: On the Interplay Between Representation Theory, Random Matrices, Special Functions, and Probability},
        publisher={World Scientific},
        pages={129-134},
        year={2009}}

@article{Tsuji,
  author={T. Tsuji},
  title={A characterization of homogeneous self-dual cones},
  journal={Tokyo J. Math.},
  volume={5},
  pages={1-12},
  year={1982}
}

@article{PlavalaCompatible,
  author={M. Pl{\'a}vala},
  title={All measurements in a probabilistic theory are compatible if and only if the state space is a simplex},
  journal={Phys. Rev. A},
  volume={94},
  pages={042108},
  year={2016},
note={{\tt arxiv:1608.05614}}
}

@article{AlfsenShultzOrientation,
    Author = {E.~M. Alfsen and F.~W. Shultz},
    Journal = {Proceedings of the National Academy of Sciences},
    Pages = {6596--6601},
    Title = {Orientation in operator algebras},
    Volume = {95},
    Year = {1998}}

@unpublished{BBLW06Cloning,
    author={H. Barnum and J. Barrett and M. Leifer and A. Wilce},
    title={Cloning and Broadcasting in Generic Probabilistic Models},
    note = {arxiv.org e-print \href{http://arxiv.org/abs/0611295}{\tt quant-ph/0611295}},
    year={2006}}

@inproceedings{NayakRandomAccess,
  author    = {Ashwin Nayak},
  title     = {Optimal Lower Bounds for Quantum Automata
               and Random Access Codes},
  booktitle = {Proceedings of the 40th Annual Symposium on
               Foundations of Computer Science},
  pages     = {369--376},
  year      = {1999},
  doi       = {10.1109/SFFCS.1999.814608},
  eprint    = {quant-ph/9904093},
  archivePrefix = {arXiv},
  url       = {https://arxiv.org/abs/quant-ph/9904093}
}

@article{VerSteegWehnerRelaxed,
  author  = {Greg Ver Steeg and Stephanie Wehner},
  title   = {Relaxed Uncertainty Relations and Information Processing},
  journal = {Quantum Information and Computation},
  volume  = {9},
  number  = {9--10},
  pages   = {801--832},
  year    = {2009},
  doi     = {10.26421/QIC9.9-10-6},
  eprint  = {0811.3771},
  archivePrefix = {arXiv},
  primaryClass  = {quant-ph},
  url     = {https://arxiv.org/abs/0811.3771}
}

@inproceedings{WolfWullschlegerOblivious,
  author    = {Stefan Wolf and J{\"u}rg Wullschleger},
  title     = {Oblivious Transfer and Quantum Non-Locality},
  booktitle = {Proceedings of the 2005 IEEE International
               Symposium on Information Theory},
  pages     = {1745--1748},
  year      = {2005},
  eprint    = {quant-ph/0502030},
  archivePrefix = {arXiv},
  url       = {https://arxiv.org/abs/quant-ph/0502030}
}

@article{PawlowskiInformationCausality,
  author  = {Marcin Paw{\l}owski and Tomasz Paterek and
             Dagomir Kaszlikowski and Valerio Scarani and
             Andreas Winter and Marek {\.Z}ukowski},
  title   = {Information Causality as a Physical Principle},
  journal = {Nature},
  volume  = {461},
  pages   = {1101--1104},
  year    = {2009},
  doi     = {10.1038/nature08400},
  eprint  = {0905.2292},
  archivePrefix = {arXiv},
  primaryClass  = {quant-ph},
  url     = {https://arxiv.org/abs/0905.2292}
}

@article{VanDamImplausible,
  author  = {Wim van Dam},
  title   = {Implausible Consequences of Superstrong Nonlocality},
  journal = {Natural Computing},
  volume  = {12},
  number  = {1},
  pages   = {9--12},
  year    = {2013},
  doi     = {10.1007/s11047-012-9353-6},
  eprint  = {quant-ph/0501159},
  archivePrefix = {arXiv},
  url     = {https://arxiv.org/abs/quant-ph/0501159}
}

@article{CavalcantiPostQuantumSteering,
  author        = {Cavalcanti, Paulo J. and Selby, John H. and Sikora, Jamie
                   and Galley, Thomas D. and Sainz, Ana Bel{\'e}n},
  title         = {Post-quantum steering is a stronger-than-quantum resource
                   for information processing},
  journal       = {npj Quantum Information},
  volume        = {8},
  number        = {1},
  pages         = {76},
  year          = {2022},
  doi           = {10.1038/s41534-022-00574-8},
  url           = {https://doi.org/10.1038/s41534-022-00574-8},
  eprint        = {2102.06581},
  archivePrefix = {arXiv},
  primaryClass  = {quant-ph}
}

@article{Deutsch96a,
  author =       "D. Deutsch and A. Ekert and R. Jozsa and C. Macchiavello and S. Popescu and A. Sanpera",
  title =        "Quantum privacy amplification and the security of quantum cryptography over noisy channels",
  journal =    "Phys Rev Lett",
  year =         "1996",
  pages =        "2818--2821",
  volume =      "77",
  note =         ""
}

@article{Lo99a,
  author =       "H.-K. Lo and H.~F. Chau",
  title =        "Unconditional security of quantum key distribution over arbitrarily long distances",
  journal =    "Science",
  year =         "1999",
  pages =        "2050--2056",
  volume =      "283",
  note =         ""
}

@Inproceedings{BGWLTShadow,
  author    = {Barnum, Howard and Graydon, Matthew A. and Wilce, Alex},
  year      = {2023},
  title     = {Locally Tomographic Shadows (Extended Abstract)},
  editor    = {Mansfield, Shane and Valiron, Beno\^it and Zamdzhiev, Vladimir},
  booktitle = {{Proceedings of the Twentieth International Conference on}
               Quantum Physics and Logic,
               {Paris, France, 17-21st July 2023}},
  series    = {Electronic Proceedings in Theoretical Computer Science},
  volume    = {384},
  publisher = {Open Publishing Association},
  pages     = {47-57},
  doi       = {10.4204/EPTCS.384.3}
}

@article{DArianoBitCommitment,
                author={G.~M. D'Ariano and D. Kretschmann and  D. Schlingemann and R.F. Werner},
                title={Quantum bit commitment: the possible and the impossible},
                journal={Physical Review A},
                volume={76},
                pages={032328},
                year={2007}}

@article{Lo97a,
                author={H.-K. Lo and H.~F. Chau},
                title={Is quantum bit-commitment really possible?},
                journal={Phys. Rev. Lett.},
                volume={78},
                pages={3410--3413},
                year={1997}}

@article{Lo98a,
                author={H.-K. Lo and H.~F. Chau},
                title={Why quantum bit commitment and ideal quantum coin tossing are impossible},
                journal={Physica D},
                volume={120},
                pages={177--187},
                year={1999}}

@article{Mayers97a,
                author={D. Mayers},
                title={Unconditionally secure quantum bit commitment is impossible},
                journal={Physical Review Letters},
                volume={78},
                pages={3413},
                year={1997}}

@INPROCEEDINGS{BBLW08,
	EDITOR    	= "S.\ Abramsky and M.\ Mislove",
 	AUTHOR    	= "H. Barnum and J. Barrett and M. Leifer and A. Wilce",
    	TITLE     	= "Teleportation in general probabilistic theories",
    	BOOKTITLE 	= "Mathematical Foundations of Information Flow (Proceedings of the Clifford Lectures 2008)",
	PUBLISHER   	= "American Mathematical Society",
  	SERIES         	= "Proceedings of Symposia in Applied Mathematics",
  	VOLUME        = "71",
    	PAGES     	= "25-47",
    	YEAR      	= "2012",
	NOTE		= "Also: {\tt arXiv:0805.3553}"
}

@INPROCEEDINGS{WilceRoyalRoad,
    EDITOR      = "R Duncan and C.\ Heunen",
    AUTHOR      = "A.\ Wilce",
        TITLE       = "A royal road to quantum theory (or thereabouts): Extended abstract",
        BOOKTITLE   = "Electronic Proceedings in Theoretical Computer Science (Proceedings 12th International Workshop on Quantum Physics and Logic)",
    VOLUME      = "236",        
        PAGES       = "245-254",
    YEAR        = "2017",
}

@ARTICLE{BMU,
    	AUTHOR    	= "H.\ Barnum and M.\ P. \ M{\"u}ller and C.\ Ududec",
    	TITLE		= "Higher-order interference and single-system postulates characterizing quantum theory",
    	JOURNAL   	= "New J. Phys.",
    	VOLUME    	= "16",
    	PAGES     	= "123029",
   	YEAR 		= "2014",
	NOTE		= "Also:  {\tt arXiv:1403.4147}"
}

@ARTICLE{Edwards,
    	AUTHOR    	= "C. M.\ Edwards",
    	TITLE		= "The operational approach to algebraic quantum theory {I}",
    	JOURNAL   	= "Comm. Math. Phys.",
    	VOLUME    	= "16",
    	PAGES     	= "207-230",
   	YEAR 		= "1970"
}

@book{LudwigAx,
	author="Ludwig, G.",
	title="An Axiomatic Basis of Quantum Mechanics",
	publisher="Springer--Verlag",
	year="1985,1987",
	volume="I, II",
}

@inproceedings{BaezNoether,
        title = {Getting to the bottom of {N}oether's theorem},
        booktitle = {The Philosophy and Physics of Noether's Theorems: A Centenary Conference},
        Editor= {Bryan W. Roberts and Nicholas Teh},
archivePrefix = {arXiv},
eprint = {2006.14741},
note = {\href{http://arxiv.org/abs/2006.14741}{arXiv:2006.14741}},
    year = {2022},
      Publisher = {Cambridge University Press},
    author={Baez, John C.},
}

@article{BarnumSteer,
archivePrefix = {arXiv},
eprint = {0912.5532},
        title = {{Ensemble Steering, Weak Self-Duality, and the Structure of Probabilistic Theories}},
     	JOURNAL   	= "Foundations of Physics",
    	VOLUME    	= "43",
    	PAGES     	= "1411-1437 ",
   	YEAR 		= "2013",
    author={Barnum, H. and Gaebler, C.P. and Wilce, A.}
}

@book{vonNeumannBook,
    author = {von Neumann, J.},
    publisher = {Princeton University Press},
    title = {Mathematical Foundations of Quantum Mechanics},
    year = {1955},
}

@article{TruongTuncel,
     title = {Geometry of homogeneous convex cones, duality mapping, and optimal self-concordant barriers},
     author = {Truong, V. A. and Tun{\c{c}}el, L.},
     journal = {Math. Program. Ser. A},
     volume = {100},
     pages = {295-316},
     year = {2004},
    }

@article{GulerTuncelCharacterization,
     title = {Characterization of the barrier parameter of homogeneous convex cones},
     author = {G{\"{u}}ler, O. and Tun{\c{c}}el, L.},
     journal = {Mathematical Programming},
     volume = {81},
     pages = {55-76},
     year = {1998},
    }

@article{ChuaRelating,
     title = {Relating homogeneous cones and positive definite cones via {$T$}-algebras},
     author = {Chua, C. B.},
     journal = {SIAM J. Optim.},
     volume = {14},
     pages = {500-506},
     year = {2003},
    }

@article{ChuaPrimalDual,
     title = {A {$T$}-algebraic approach to primal-dual interior point algorithms},
     author = {Chua, C. B.},
     journal = {SIAM J. Optim.},
     volume = {20},
     pages = {503-523},
     year = {2009},
    }

@article{KoszulOuverts,
	author="J.~L. Koszul",
	title="Ouverts convexes homog{\`{e}}nes des espaces affines",
	journal="Mathematische Zeitschrift",
	year="1962",
	volume="79",
	pages="254-259",
}

@article{GiulioDer11,
archivePrefix = {arXiv},
eprint = {1011.6451},
note = {\href{http://arxiv.org/abs/1011.6451}{arXiv:1011.6451}},
  pages = {012311},
  author = {Chiribella, G. and D'Ariano, G.M. and Perinotti, P.},
  numpages = {39},
  volume = {84},
  year = {2011},
  journal = {Phys. Rev. A},
  title = {Informational derivation of quantum theory},
  doi = {10.1103/PhysRevA.84.012311},
  url = {http://link.aps.org/doi/10.1103/PhysRevA.84.012311},
  issue = {1},
  publisher = {American Physical Society},
}

@article{BarrettGenProb,
	author="Barrett, J.",
	title="Information processing in generalized probabilistic theories",
	journal="Phys. Rev. A",
	year="2007",
	volume="75",
	pages="032304",
archivePrefix = {arXiv},
eprint = {0508211},
note = {\href{http://arxiv.org/abs/quant-ph/0508211}{arXiv:0508211}},
}

@article{ConnesSelfDualCones,
	author="Connes, A.",
	title="Caract{\'e}risation des espaces vectoriels ordonn{\'e}es sous-jacents 
aux alg{\`e}bres de von {N}eumann",
	journal="Annales de l'Institut Fourier (Grenoble)",
	year="1978",
	volume="24",
	pages="121-155",
}

@article{Koecher,
 author = {Koecher, M.},
 title = {Die {G}eod\"{a}tischen von {P}ositivit\"{a}tsbereichen},
 journal = {Math. Annalen},
 volume = {135},
 year = {1958},
 pages = {192--202},
}

@article{VinbergHomogeneousCones,
 author = {E.~B. Vinberg},
 title = {Homogeneous  cones},
 journal = {Dokl. Akad. Nauk. SSSR},
 volume = {133},
 year = {1960},
 pages = {9-12},
note = {In Russian.  English translation in Soviet Math. Dokl. {\bf 1},
787-790 (1960).}
}

@article{MullerSD2012,
  publisher = {American Physical Society},
  author = {M\"{u}ller, M. P. and Ududec, C.},
  url = {http://link.aps.org/doi/10.1103/PhysRevLett.108.130401},
  year = {2012},
  pages = {130401},
  title = {The Structure of Reversible Computation Determines the Self-Duality of Quantum Theory},
  volume = {108},
  journal = {Phys. Rev. Lett.},
  doi = {10.1103/PhysRevLett.108.130401},
  numpages = {5},
  issue = {13},
  archivePrefix = {arXiv},
eprint = {1110.3516},
note = {\href{http://arxiv.org/abs/1110.3516}{arXiv:1110.3516}},
}

@article{BaHiSpectral,
	author="H. Barnum and J. Hilgert",
	title="Spectral properties of convex bodies",
	journal="J. Lie Theory",
	year="2020",
	volume="30",
    number="2",
	pages="315-344",
	archivePrefix = {arXiv},
	eprint = "1904.03753",
	note = "\href{http://arxiv.org/abs/quant-ph/1904.03753}{arXiv:1904.03753}; earlier versions titled Strongly symmetric compact convex sets are {J}ordan algebra state spaces",
}

@book{ASBook2,
	author="E.~M. Alfsen and F.~W. Shultz",
	title="Geometry of State Spaces of Operator Algebras",
	publisher="Birkh{\"a}user",
	year="2003",
}

@book{FarautKoranyi,
	author="Faraut, J. and Koranyi, A.",
	title="Analysis on {S}ymmetric Cones",
	publisher="Oxford University Press",
	year="1994",
}

@article{KrummEtAlThermo,
  author={Krumm, M. and Barnum, H. and Barrett, J. and M\"{u}ller, M. P. },
  title={Thermodynamics and the structure of quantum theory},
  journal={New J. Phys.},
  volume={19},
  pages={043025},
  year={2017},
    archivePrefix = {arXiv},
eprint = {1608.04461},
note = {\href{http://arxiv.org/abs/1608.04461}{arXiv:1608.04461}},
}

@article{ChiribellaScandoloMicrocanonical,
  author={G. Chiribella and C.~M. Scandolo},
  title={Microcanonical Thermodynamics in General Physical Theories},
  journal={New J.  Phys.},
  volume={19},
  pages={123043},
  year={2017},
    archivePrefix = {arXiv},
eprint = {1608.04459},
note = {\href{http://arxiv.org/abs/1608.04459}{arXiv:1608.04459}},
}

@article{ChiribellaScandoloEntanglementAxiomatic,
  author={G. Chiribella and C.~M. Scandolo},
  title={Entanglement as an axiomatic foundation for statistical mechanics},
  year={2016},
    archivePrefix = {arXiv},
eprint = {1608.0449},
note = {\href{http://arxiv.org/abs/1608.0449}{arXiv:1608.0449}},
}

@article{ChiribellaScandoloDiagonalization,
  author={G. Chiribella and C.~M. Scandolo},
  title={Operational axioms for diagonalizing quantum states},
  year={2015},
    archivePrefix = {arXiv},
eprint = {1506.00380},
note = {\href{http://arxiv.org/abs/1506.00380}{arXiv:1506.00380}},
}

@article{MasanesAxiom,
  author={{{\relax Ll}} Masanes and M.~P. M\"{u}ller},
  title={A derivation of quantum theory from physical requirements},
  journal={New J. Phys.},
  volume={13},
  number={6},
  pages={063001},
  year={2011},
    archivePrefix = {arXiv},
eprint = {1004.1483},
note = {\href{http://arxiv.org/abs/1004.1483}{arXiv:1004.1483}},
}

@INPROCEEDINGS{QPLthermo,
	EDITOR    	= "C.\ Heunen and P.\ Selinger and J.\ Vicary",
 	AUTHOR    	= "H.\ Barnum and J.\ Barrett and M.\ Krumm and M.\ P. \ M{\"u}ller",
    	TITLE     	= "Entropy, majorization and thermodynamics in general probabilistic theories",
    	BOOKTITLE 	= "Electronic Proceedings in Theoretical Computer Science (Proceedings 12th International Workshop on Quantum Physics and Logic)",
	VOLUME    	= "195",    	
    	PAGES     	= "43-58",
	YEAR      	= "2015",
	NOTE		= "Also: {\tt arXiv:1508.03107}"
}

@ARTICLE{JNW,
    	AUTHOR    	= "P.\ Jordan and J.\ von Neumann and E.\ Wigner",
    	TITLE		= "On an algebraic generalization of the quantum mechanical formalism",
    	JOURNAL   	= "Ann. Math.",
    	VOLUME    	= "35",
    	PAGES     	= "29-64",
   	YEAR 		= "1934"
}

@ARTICLE{VinbergHomogeneous,
    	AUTHOR    	= "E.~B. Vinberg",
    	TITLE     	= "Homogeneous cones",
    	JOURNAL   	= "Dokl. Acad. Nauk. SSSR",
    	VOLUME    	= "133",
    	PAGES     	= "9-12",
	YEAR 		= "1960",
	NOTE      	= "English translation: E. Vinberg (1961): Soviet Math. Dokl. 1, pp. 787-790"
}

\begin{appendix}

\section{Proofs}

\subsection{Proofs concerning steering}

We remind the reader that the definitions of a bipartite state steering its marginal, uniform steering, universal steering, and universal self-steering are given in Section~\ref{sec:steering}.

Our goal here is to prove that uniform steering implies homogeneity. To do so, we will use results from~\cite{BarnumSteer}. However, they assume local tomography, which we don't, and we include a condition that they omitted in their Proposition 5.6 because, 
although we do not know whether the proposition is true without it, it appears necessary  to complete the argument in \cite{BarnumSteer}. 

It is convenient to represent probabilistic ensembles $p_i, \omega_i$ 
of normalised states of a system as sets of unnormalised states $\omega_i$, each of which can be interpreted as a normalised state $\reallywidetilde{\omega_i} := \omega_i/u(\omega_i)$ times a probability $p_i = u(\omega_i)$ (here we write $u(\omega_i) := \inner{u,\omega}$ as a shorthand for the probability function).  The average state is then $\omega := \sum_i p_i \reallywidetilde{\omega_i} = \sum_i \omega_i,$ and we say $\omega_i$ are
\emph{an ensemble for} $ \omega$.     

\begin{definition}
    Let $A$ and $B$ be systems of a GPT and let $\omega_{AB}$ be some bipartite state. Its associated \emph{conditioning map} $\hat{\omega}: V_A^* \rightarrow V_B$ is defined by the condition that for all effects $f \in V_B^*$, $\hat{\omega}(e)(f) := \omega(e \otimes f)$ (recall that a state is fully defined by its probabilities on all effects by the final condition in Definition~\ref{def:GPT}). 
    This function hence takes effects $e$ to unnormalised states $\hat{\omega}(e)$, which we also write as $\omega_B^e$.
    It is a positive linear map with $\homega(u_A) = \omega_B$ equal to the $B$ marginal of $\omega$. 
\end{definition}
 
For \emph{any}
measurement $\{e^i_A\}$ on $A$, the unnormalised
conditional states $\omega_B^{e_i}$ are an ensemble for $\omega_B$.  As this
suggests, the map $\homega$ is useful for studying steering.  Indeed, the following is
essentially a restatement of the definition of a bipartite state $\omega$'s steering its marginal, in terms of the map $\homega$.

\begin{proposition}\label{prop: steering reformulated}
$\omega^{AB}$ is steering for its marginal $\omega_B$ iff for every ensemble 
$\omega_B^i$ for $\omega^B$, there exists a measurement $e^i_A$ on $A$ such 
that $\homega(e_A^i) = \omega_B^i$.
\end{proposition}

While local tomography was a standing assumption in \cite{BarnumSteer}, the discussion of
steering so far made no use of it.  Neither do the proofs of the following results
on steering.
This is because even in a non-locally-tomographic context, it is the conditioning map $\hat{\omega}: A^* \rightarrow B$ induced by the state $\omega$ of $AB$ that matters for steering, and this map is determined entirely by (indeed bijectively corresponds to) the \emph{locally tomographic reduced state} 
(aka \emph{locally tomographic shadow}~\cite{BGWLTShadow}) 
of the bipartite state $\omega$. 
This is the element of $(V_A^* \otimes V_B^*)^* \,(\,\cong V_A \otimes V_B)$ determined by $\omega$'s values on elements $e \otimes f$, e.g. product effects.  

\begin{lemma}[Lemma 5.2 in~\cite{BarnumSteer}]\label{lemma: steering surjective onto face}
If $\omega \in AB$ steers its marginal $\omega_B$, then 
$\homega(A^*_+) = \face{(\omega_B)}$ (recall Definition~\ref{def:face}).
\end{lemma}

Indeed, if $\omega$ steers its $B$-marginal, $\homega$ must act surjectively from $[0,u_A]$ onto the order-interval $[0, \omega_B]$, since the unnormalised states in $[0,\omega^B]$ are precisely the subnormalised states that may appear in ensembles for $\omega^B$
(when these ensembles are written, as we have been doing, 
as lists of subnormalised states). 
In fact, we have: 

\begin{theorem}[Theorem 5 in \cite{BarnumSteer}]
\label{theorem: strong quotient map}
$\omega \in AB$ is steering for its $B$ marginal iff $\homega\downharpoonleft_{[0,u_A]} : [0,u_A] \rightarrow [0,\omega_B]$ is a strong quotient map of ordered sets.
\end{theorem}
$\homega\downharpoonleft_{[0,u_A]}$ is just the restriction of $\homega$ to $[0,u_A]$, corestricted to $[0, \omega_B]$.  (If we have unrestricted effects, these will be its domain and range.)  
A strong quotient map of ordered sets $\phi: X \rightarrow Y$ is one such that each finite chain 
in $Y$ is the image of some chain in $X$.  The equivalence follows from 
Proposition \ref{prop: steering reformulated} by considering the elements
of the measurement $M = e_1,...,e_n$ as the differences of elements in a chain 
$c_1 \le c_2 \le \cdots c_k \cdots \le c_n$, with $c_k :=\sum_{i=1}^k e_k$, and 
similarly for the elements of the ensemble $\omega_B^i$; see \cite{BarnumSteer} for the details.

{}

\begin{corollary}[Corollary 5.5 in \cite{BarnumSteer}] 
\label{cor: injective interior steering}
Let $\omega$ be steering for $\omega_B \in B^\circ_+$, the interior of $B_+$, so that $\face{(\omega_B)} = B_+$.
If $\homega$ is injective, then it is an order isomorphism.  
\end{corollary}
In this situation, $A$ must have unrestricted effects, for $[0, u_A]$ must be
isomorphic to the full o{}rder interval $[0,\omega^B]$.

Given this, we can now state a corrected and generalised version of Proposition 5.6 from \cite{BarnumSteer}:
\begin{theorem*}[Restatement of Theorem \ref{prop: universal steering implies homogeneity}]
In any theory that supports universal uniform steering  where the steering system $A$ is of the same dimension as the steered system  $B$, every
 irreducible finite-dimensional state space in the theory is homogeneous.
\end{theorem*}

The correction (needed for the argument in \cite{BarnumSteer}, at least)   
is the added dimensionality requirement, since a linear surjection of a vector space $V$ onto $W$ of the same dimension must be a linear bijection.  So $\homega$ is injective, implying, by Corollary \ref{cor: injective interior steering}, that it is an order isomorphism.

\subsection{Proof of derivation of quantum theory from local tomography (Theorem
\ref{theorem: local tomography 1})}\label{app:local-tomography}

As discussed in the sketch of the proof, we may assume that our system is a simple EJA.  We furthermore have locally tomographic composites.
Lemmas \ref{lem:pure-effects} and \ref{lem:pure-states} below apply to arbitrary locally tomographic composites with 
unrestricted states and effects.  Lemma \ref{lem: pure effects ray-extremal}
(at least as proved here) applies to the full  
interval of effects $E = [0,u]$ in a cone generated by effects, 
and Lemma \ref{lem:pure-effects} applies to a situation where $E_A$ and $E_B$ are
such full order intervals.  But the use of Lemma \ref{lem:pure-effects} as 
a mathematical tool to obtain Theorem \ref{theorem: local tomography 1} does not
impose unrestricted effects on the GPT systems in that theorem.  

\begin{lemma}[Lemma 3 of~\cite{BBLW06Cloning}]\label{lem:pure-marginal-state}
    If either marginal of a bipartite state $\omega_{AB}$ is pure then the state is a product state $\omega_A\otimes \omega_B$.
\end{lemma}
This lemma as stated in \cite{BBLW06Cloning} applied to bipartite states in the maximal tensor product, but continues to hold for any composite system whose state cone is contained in the maximal tensor product. That is the case in our setting since we are assuming local tomography and unrestricted effects.

Using this lemma we can prove that a composite of pure states is again pure.
\begin{lemma}\label{lem:pure-states}
    Let $\omega_A$ and $\omega_B$ be pure states. Then $\omega_A\otimes \omega_B$ is also pure.
\end{lemma}
\begin{proof}
    Suppose $\omega_{AB} := \omega_A\otimes \omega_B$ with $\omega_A$ and $\omega_B$ pure can be decomposed as $\omega_{AB} = \lambda \sigma_{AB} + (1-\lambda) \sigma_{AB}'$ for some $0<\lambda<1$ and bipartite states $\sigma_{AB}$ and $\sigma_{AB}'$. We need to show that $\sigma_{AB} = \sigma_{AB}' = \omega_A\otimes \omega_B$.

    The $A$ marginal of $\omega_{AB}$ is $\omega_A$ which is pure. By assumption it is also equal to $\lambda \sigma_A + (1-\lambda) \sigma_A'$. Hence, $\sigma_A = \sigma_A' = \omega_A$ are pure. By the previous lemma, $\sigma_{AB}$ and $\sigma_{AB}'$ are then product states. We can do a similar calculation with the $B$ marginal, so that we end up with $\sigma_{AB} = \sigma_{AB}' = \omega_A\otimes \omega_B$.
\end{proof}

\begin{lemma}[Corollary of Lemma \ref{lem:pure-states}]
\label{lem: tensoring ray-extremal}
Let $x$ and $y$ be ray-extremal elements of the positive cones $V^A_+$ and
$V^B_+$ respectively, and let $V^A_+ \omin V^B_+ \subseteq V^{AB}_+ \subseteq V^A_+ \omax V^B_+$.  Then $x \otimes y$ is ray-extremal in $V^{AB}_+$.  
\end{lemma}
\begin{proof}
The conclusion is trivial if either of $x,y$ are zero.  Suppose neither is.  
Let $\Omega^A, \Omega^B$ be any bases for $V^A_+, V^B_+$, determined by interior
elements $u_A$ and $u_B$ in $(V^A_+)^*, (V^B_+)^*$ respectively, and let 
$\Omega^{AB}$ be the base $\{\omega \in V^{AB}_+: (u_A \otimes u_B)(\omega) = 1\}$ 
Since they are (nonzero) ray-extremal, $x = \lambda_x \omega_x, y = \lambda_y 
\omega_y$
for unique strictly positive $\lambda_x, \lambda_y$ and unique pure $\omega_x \in \Omega^A, \omega_y \in \Omega^B$.  By Lemma \ref{lem:pure-states}, 
$\omega_x \otimes \omega_y$ is pure in $\Omega^{AB}$.  But 
$x \otimes y = \lambda_x \lambda_y (\omega_x \otimes \omega_y)$, a strictly 
positive multiple of a pure state of $\Omega^{AB}$.   Consequently it is 
ray-extremal in $V^{AB}_+$.  
\end{proof}

Recall that a pure effect was defined as one that cannot be expressed as a sum
of nonzero effects not proportional to each other.
\begin{lemma}\label{lem: pure effects ray-extremal}
Let the set of effects be the order-interval $[0, u]$ in the cone generated
by the effects. Then purity of an effect is equivalent to its ray extremality.
\end{lemma}
\begin{proof}
Suppose an effect to be ray-extremal in the cone generated by the 
set of effects.  It is immediate from the definition of face and 
of extremal ray as a one-dimensional face (Definition~\ref{def:face}) that it cannot be expressed as a sum of effects unless they lie in that ray.  On the 
other hand, suppose an effect $f$ is not ray extremal.  Then it can be expressed
as a sum $f = a + b$ of nonzero elements in the cone at least one of which is not proportional to that effect.  But $a,b$ are each effects because  
$a,b > 0$ so $0 < b= f - a \le u - a \le u$ (and similarly with $b$ and $a$ interchanged).  Hence $f$ is not pure. 
\end{proof}

The following is then a case of Lemma \ref{lem: tensoring ray-extremal}. 
\begin{lemma}\label{lem:pure-effects}
    Let $e_A$ and $e_B$ be pure effects. Then $e_A\otimes e_B$ is pure.
\end{lemma}

\begin{theorem*}[Restatement of Theorem \ref{theorem: local tomography 1}]
    Let $A$ be a system in a GPT where composites are locally tomographic and where every system satisfies homogeneity and continuous pure transitivity. 
    Then $V_A$ is a simple complex matrix algebra.
\end{theorem*}
\begin{proof}[Proof of Theorem~\ref{theorem: local tomography 1}]
By Corollary~\ref{cor:homogeneity-simple}, $V^A_+$ is the positive cone of a simple Jordan
algebra with unit $u$ and rank $r$.    
Consider the system with the same state space and unrestricted
effects $E = \{e: 0 \le e \le u\}$ where the ordering is that determined
by $(V^A_+)^*$.  
Let $\omega_1,\ldots, \omega_r$ be a frame of pure states (i.e.~a maximal collection of orthogonal pure states) of the system $A$ and let $e_1,\ldots, e_r$ be its corresponding set of effects pure in $E$, which satisfy $\inn{\omega_i,e_j} = \delta_{ij}$ (which exist because we are working in a simple EJA).
Note that $\inn{\omega_i\otimes \omega_j,e_k\otimes e_l} = \inn{\omega_i,e_k}\inn{\omega_j,e_l} = \delta_{ik}\delta_{jl}$ so that the $\omega_i\otimes \omega_j$ are still perfectly distinguishable. As $\sum_{i,j} e_i\otimes e_j = u\otimes u = u_{AA}$, and each of the $e_i\otimes e_j$ is pure, this is also a maximal collection of such states, so that we know that the rank of $AA$ is $r^2$.  The ranks and dimensions of all simple EJAs
are known---see e.g. \cite{FarautKoranyi}, p. 97.  
We can now make a case distinction on the types of simple Jordan algebras to see that $V_A$ can only be complex. 
For instance, if $V_A$ is the rank-3 Albert algebra of $3\times 3$ self-adjoint octonion matrices, which is 27-dimensional, then the composite $AA$ would need to be a rank 9 simple algebra of dimension $27^2=729$ (due to local tomography). But such a simple algebra doesn't exist. Quaternionic systems of rank at least 3 and spin-factors of dimension $>4$ can be ruled out in a similar way. 
This leaves real and complex systems as the only option. However, for rank-$r$ real systems $A$, whose dimension is $d = r(r+1)/2$, it is straightforward to check that the 
dimension $d^2 = r^2(r+1)^2/4$ of a locally tomographic composite $AA$ is smaller than 
that of a 
rank-$r^2$ real or complex composite. In contrast, for complex systems the relation $d = r^2$ is consistent with locally tomographic composites, which we know exist.
\end{proof}

Note that this proof only works because we are dealing with simple algebras (due to our assumption of continuous pure transitivity).
For non-simple algebras, just dimension and rank counting is not enough: if we tensor a direct sum of 3 Albert algebras with itself, its composite should have rank 81 and dimension $81^2$. But these are exactly the parameters of $M_{81}(\C)_{\text{sa}}$. It is likely there are other ways to rule out these possibilities, but we don't know how to do it without making additional assumptions. We will see one way to do this in Theorem~\ref{thm:non-simple-algebras}.

\section{Variations on ruling out non-complex Jordan-algebraic systems}
\label{app:variations}

Our result characterising quantum systems used various additional assumptions that we can vary: we could require just pure transitivity instead of continuous pure transitivity, or drop the local tomography requirement. 
When we do so, we add some other assumptions in order to obtain interesting
results.

First, a result where we don't require local tomography, at the cost of explicitly requiring the composite to preserve purity of states.
\begin{theorem*}[Restatement of Theorem~\ref{thm:real-and-complex}]
    Let $A$ be a system in a GPT satisfying homogeneity and continuous pure transitivity, with unrestricted effects, 
    and suppose that composites preserve 
    purity (i.e.~a composite of pure states or pure effects is pure).
    Then $V_A$ is either a complex or real matrix algebra.
\end{theorem*}
\begin{proof}
    Since we have continuous pure transitivity, $V_A$ is a simple Jordan algebra. Since we don't have local tomography, we can't use Lemmas~\ref{lem:pure-states} and~\ref{lem:pure-effects}, which is why we explicitly assume that composites preserve pure states and
    effects.  We again consider  a frame
    $\omega_i, i \in \{1,..r\}$ and corresponding pure effects $e_j$ that perfectly 
    distinguish them.  
    We can then proceed in the same way as in the proof of Theorem~\ref{theorem: local tomography 1}, reducing the problem to dimension- and rank-counting. 
    But now, since we don't have local tomography, the dimension restrictions are less severe, requiring only that $\dim V_{AA} \geq (\dim V_A)^2$, instead of strict equality. This however still allows us to rule out most of the possibilities.

    For instance, if $V_A$ is the rank-3 Albert algebra of $3\times 3$ self-adjoint octonion matrices, which is 27-dimensional, then the composite would need to be a rank 9 simple algebra of dimension at least $27^2=729$ (because composites are injective). But such a simple algebra doesn't exist (we need only check the three cases $M_9(\R), M_9(\C)$ and  $M_9(\H)$).  Quaternionic systems of rank at least 3 and spin-factors of dimension $>4$ can be ruled out in a similar way. This then leaves the real and complex systems as options, as desired.
\end{proof}

We can weaken the transitivity assumption from continuous pure transitivity to just pure
transitivity to allow non-simple algebras. To single out real and complex systems in this case, 
we use an additional assumption on composites, which concerns \emph{classical effects}.  These are effects $e$ which for any pure state $\omega$ give $\inn{\omega ,e} = 1$ or $\inn{\omega, e} = 0$. That is, they hold with certainty, or they do not hold at all. In a classical GPT all pure effects are classical, while in an irreducible quantum system the only classical effect is the trace.  In a Jordan algebra the classical effects precisely correspond to the different factors of the Jordan algebra. We say a composite is \emph{classicality preserving} when a composite of classical effects is again classical.
\begin{theorem*}[Restatement of Theorem  \ref{thm:non-simple-algebras}]
    Let $A$ be a system in a GPT satisfying homogeneity and pure transitivity, and having unrestricted effects. Suppose furthermore that composites are purity preserving 
    (on both states and effects)
    and classicality preserving.
    Then $V_A$ is a direct sum of complex or real matrix algebras.
\end{theorem*}
\begin{proof}
    Since $A$ satisfies homogeneity and pure transitivity, $V_A$ must be a direct sum of copies of the same simple Jordan algebra $W$. Let $\omega_1,\ldots,\omega_r$ be a frame of $W$ and let $e_1,\ldots, e_r$ be the corresponding pure effects such that $\inn{\omega_i,e_j} = \delta_{ij}$.
    A property of simple Jordan algebras is that for every frame we can find a pure state that has non-zero overlap with all the states in the frame. That is, there exists a pure state $\omega$ and corresponding pure effect $e$, such that $\inn{\omega_i,e} \neq 0$ for all $i$. For instance, if $W=M_r(\C)_{\text{sa}}$ and $\omega_i = \ket{i}$, then we can take $\omega = \frac{1}{\sqrt{r}}\sum_i \ket{i}$.

    For any fixed copy of $W$ in $A$, all these states $\omega_i$ and effects $e_j$ correspond to pure states and effects of $A$.  Let $W$ be an arbitrary such copy.   Then 
    $u_W = \sum_i e_i$ is a classical effect in $A$, because it corresponds to one of its simple summands.
    Consider the composite $A\otimes A$ and its associated Jordan algebra $V_{AA}$. Because we assume purity preserving composites, the states $\omega_i\otimes \omega_j$ and effects $e_i\otimes e_j$ are pure. We see that $\inn{\omega_i\otimes \omega_j, e\otimes e} = \inn{\omega_i,e}\inn{\omega_j,e} \neq 0$. There is hence a pure effect $e\otimes e$ that has non-zero overlap with all the pure states $\omega_i\otimes \omega_j$. This is only possible if all $\omega_i\otimes \omega_j$ belong to the same simple factor of the Jordan algebra $V_{AA}$. Let's call this factor $E$.
    Since $\inn{\omega_i\otimes \omega_j,e_k\otimes e_l} = \inn{\omega_i,e_k}\inn{\omega_j,e_l} = \delta_{ik}\delta_{jl}$, these pure states are still orthogonal, so that the rank of $E$ is at least $r^2$.

    All the pure effects $e_i\otimes e_j$ also belong to $E$.  Since $\sum_{ij} e_i\otimes e_j = u_W\otimes u_W$ is a tensor product of classical effects, this is then also classical, which means it must be the unit effect of this simple summand (since the unit is the only classical effect in a simple Jordan algebra). The rank of $E$ is hence $r^2$. Furthermore, because the tensor product is injective, and all of $W\otimes W$ is mapped into $E$, we must have $(\dim W)^2 \leq \dim E$. With the same argument as  in the proof of Theorem \ref{thm:real-and-complex} we can then rule out all possibilities for $W$ except  the real and complex matrix algebras.
\end{proof}

We may also obtain the same conclusion as in 
Theorem \ref{thm:non-simple-algebras} if instead of classicality preservation
for effects, we make a classicality preservation assumption 
on the state space: that the composition of state spaces 
commutes with direct sum decompositions of them (as ordered 
linear spaces): 
$V^A = V^A_1 \oplus V^A_2 \implies V^{AB} = V^{A_1 B} \oplus V^{A_2 B}$, where $V^{A_i B}$ is the state space of a
 composite of $A_i$ and $B$, and similarly for direct sum decompositions
 of $B$ instead of $A$.  Unrestricted effects will then give us the classical effects that distinguish these summands, along with classicality preservation for effects.  
 
Such classicality preservation assumptions are not redundant as they do not automatically follow from the notions of system and composite system we have been using.  When they do not hold, we can have states
that exhibit what one might call ``nonlocal coherence''.

\end{appendix}

\end{document}